%% file: paper.tex
\documentclass[conference]{IEEEtran}
\IEEEoverridecommandlockouts

\usepackage{cite}
\ifCLASSINFOpdf
  \usepackage[pdftex]{graphicx}
  \graphicspath{{figures/}}
  \DeclareGraphicsExtensions{.pdf,.jpeg,.png}
\else
\fi

\usepackage{amsmath}
\usepackage{amsthm}

\usepackage{algorithmic}

\usepackage{array}
\usepackage{url}
\usepackage{verbatim}

\newtheorem{theorem}{Theorem}
\newtheorem{lemma}{Lemma}[theorem]

\usepackage{tikz}

\usetikzlibrary{backgrounds}
\usetikzlibrary{calc}
\usetikzlibrary{positioning}
\usetikzlibrary{intersections}

\usepackage{pgfplots}
\pgfplotsset{compat=1.9}

\definecolor{plotcolor0}{RGB}{255,0,0}
\definecolor{plotcolor1}{RGB}{0,0,255}
\definecolor{plotcolor2}{RGB}{64,170,0}
\definecolor{plotcolor3}{RGB}{170,0,255}
\definecolor{plotcolor4}{RGB}{0,170,255}
\definecolor{plotcolor5}{RGB}{164,164,0}
\definecolor{plotcolor6}{RGB}{128,128,128}
\definecolor{plotcolor7}{RGB}{0,0,0}
\definecolor{plotcolor8}{RGB}{0,110,165}
\definecolor{plotcolor9}{RGB}{128,64,0}

\pgfplotscreateplotcyclelist{my_color}{%
  plotcolor0, every mark/.append style={solid,scale=0.8}, mark=x \\%
  plotcolor1, every mark/.append style={solid,scale=0.6}, mark=o \\%
  plotcolor2, every mark/.append style={solid,scale=0.8}, mark=diamond \\%
  plotcolor3, every mark/.append style={solid,scale=0.8}, mark=triangle \\%
  plotcolor4, every mark/.append style={solid,scale=0.8,rotate=180}, mark=triangle \\%
  plotcolor5, every mark/.append style={solid,scale=0.65}, mark=square \\%
  plotcolor6, every mark/.append style={solid,scale=0.8}, mark=star \\%
  plotcolor7, every mark/.append style={solid,scale=0.8}, mark=+ \\%
  plotcolor8, every mark/.append style={solid,scale=0.8}, mark=pentagon \\%
  plotcolor9, every mark/.append style={solid,scale=0.9}, mark=Mercedes star flipped \\%
  plotcolor0, every mark/.append style={solid,scale=0.8}, mark=x, dashed \\%
  plotcolor1, every mark/.append style={solid,scale=0.6}, mark=o, dashed \\%
  plotcolor2, every mark/.append style={solid,scale=0.8}, mark=diamond, dashed \\%
  plotcolor3, every mark/.append style={solid,scale=0.8}, mark=triangle, dashed \\%
  plotcolor4, every mark/.append style={solid,scale=0.8,rotate=180}, mark=triangle, dashed \\%
  plotcolor5, every mark/.append style={solid,scale=0.65}, mark=square, dashed \\%
  plotcolor6, every mark/.append style={solid,scale=0.8}, mark=star, dashed \\%
  plotcolor7, every mark/.append style={solid,scale=0.8}, mark=+, dashed \\%
  plotcolor8, every mark/.append style={solid,scale=0.8}, mark=pentagon, dashed \\%
  plotcolor9, every mark/.append style={solid,scale=0.9}, mark=Mercedes star flipped, dashed \\%
  plotcolor0, every mark/.append style={solid,scale=0.8}, mark=x, densely dotted \\%
  plotcolor1, every mark/.append style={solid,scale=0.6}, mark=o, densely dotted \\%
  plotcolor2, every mark/.append style={solid,scale=0.8}, mark=diamond, densely dotted \\%
  plotcolor3, every mark/.append style={solid,scale=0.8}, mark=triangle, densely dotted \\%
  plotcolor4, every mark/.append style={solid,scale=0.8,rotate=180}, mark=triangle, densely dotted \\%
  plotcolor5, every mark/.append style={solid,scale=0.65}, mark=square, densely dotted \\%
  plotcolor6, every mark/.append style={solid,scale=0.8}, mark=star, densely dotted \\%
  plotcolor7, every mark/.append style={solid,scale=0.8}, mark=+, densely dotted \\%
  plotcolor8, every mark/.append style={solid,scale=0.8}, mark=pentagon, densely dotted \\%
  plotcolor9, every mark/.append style={solid,scale=0.9}, mark=Mercedes star flipped, densely dotted \\%
}

\pgfplotsset{
  myPlot/.style={
    grid,
    width=59mm,height=48mm,
    major grid style={thin,dotted,color=black!50},
    minor grid style={thin,dotted,color=black!50},
    every axis/.append style={
      thin,
      tick style={
        line cap=round,
        thin,
      },
    },
    cycle list name={my_color},
    major tick length=3pt,
    minor tick length=1.5pt,
    legend cell align=left,
    xlabel near ticks,
    ylabel near ticks,
    title style={yshift=-0.6ex,font=\small},
  },
}

\pgfplotsset{
  log x ticks with fixed point/.style={
    xticklabel={
      \pgfkeys{/pgf/fpu=true}
      \pgfmathparse{exp(\tick)}%
      \pgfmathprintnumber[fixed relative, precision=3]{\pgfmathresult}
      \pgfkeys{/pgf/fpu=false}
    }
  },
  log y ticks with fixed point/.style={
    yticklabel={
      \pgfkeys{/pgf/fpu=true}
      \pgfmathparse{exp(\tick)}%
      \pgfmathprintnumber[fixed relative, precision=3]{\pgfmathresult}
      \pgfkeys{/pgf/fpu=false}
    }
  }
}

\newcommand{\coloneqq}{:=}

\usepackage{csquotes}
\usepackage{color}
\usepackage{etoolbox}
\usepackage{xspace}
\usepackage[left,pagewise]{lineno}
\definecolor {infocolor} {rgb} {0.6,0.6,0.6}

\usepackage{microtype}

\input{makros.tex}

\newcommand{\StringSet}[1]{\mathcal{#1}}

\newcommand{\dpre}[1]{\textsc{dist}{(#1)}}

\newcommand{\lcpMathh}[2]{\textsc{lcp(}#1, #2\textsc{)}}

\newcommand{\CommonCrawl}{\textsc{CommonCrawl}\xspace}
\newcommand{\DnaReads}{\textsc{DnaReads}\xspace}

\newcommand{\dton}{$D/N$\xspace}

\newcommand{\MSLCPS}{\texttt{MS}\xspace}

\newcommand{\MSSimpleS}{\texttt{MS-simple}\xspace}
\newcommand{\PDGolombS}{\texttt{PDMS-Golomb}\xspace}
\newcommand{\PDNoGolombS}{\texttt{PDMS}\xspace}

\newcommand{\kurpicz}{\texttt{FKmerge}\xspace}
\newcommand{\hQuick}{\texttt{hQuick}\xspace}

\newcommand{\dMSS}{MS\xspace}
\newcommand{\dPDSS}{PDMS\xspace}
\newcommand{\lmax}{\hat{\ell}}
\newcommand{\Dmax}{\hat{D}}
\newcommand{\dmax}{\hat{d}}
\newcommand{\nmax}{\hat{n}}
\newcommand{\Nmax}{\hat{N}}

\begin{document}
\pagestyle{plain}
\title{Communication-Efficient String Sorting}


\author{%
  \IEEEauthorblockN{Timo Bingmann, Peter Sanders, Matthias Schimek}
  \IEEEauthorblockA{\textit{Karlsruhe Institute of Technology}, \textit{Karlsruhe, Germany} \\
  \small\textit{\{\url{bingmann, sanders}\}\url{@kit.edu}, \url{matthias_schimek@gmx.de}}  \\
}}


%Telephone: (800) 555--1212, Fax: (888) 555--1212}


\maketitle

\begin{abstract}
There has been surprisingly little work on algorithms for sorting strings on distributed-memory parallel machines. We develop efficient algorithms for this problem based on the multi-way merging principle. These algorithms inspect only characters that are needed to determine the sorting order. Moreover, communication volume is reduced by also communicating (roughly) only those characters and by communicating repetitions of the same prefixes only once. Experiments on up to 1280 cores reveal that these algorithm are often more than five times faster than previous algorithms.
\end{abstract}

\begin{IEEEkeywords}
distributed-memory algorithm, string sorting, communication-efficient algorithm
\end{IEEEkeywords}

% \fi
\IEEEpeerreviewmaketitle

%%%%%%%%%%%%%%%%%%%%%%%%%%%%%%%%%%%%%%%%%%%%%%%%%%%%%%%%%%%%%%%%%%%%%%
\section{Introduction}
Sorting, i.e., establishing the global ordering of $n$ elements $s_0$,
\ldots $s_{n-1}$, is one of the most fundamental and most frequently
used subroutines in computer programs. For example, sorting is used
for building index data structures like B-trees, inverted indices or
suffix arrays, or for bringing data together that needs to be
processed together. Often, the elements have string keys, i.e.,
variable length sequences of characters, or, more generally, multiple
subcomponents that are sorted lexicographically.
For example, this is the case for sorted arrays of strings that facilitate
fast binary search, for prefix B-trees \cite{bayer1977prefix,GraLar01}, or when using
string sorting as a subroutine for suffix sorting (i.e., the problem of sorting all suffixes of \emph{one} string). Using string sorting for suffix sorting can mean to directly sort the suffixes \cite{futamura2001parallel}, or to sort shorter strings as a subroutine. For example, the difference cover algorithm \cite{KarSan05} is theoretically one of the most scalable suffix sorting algorithms. An implementation with large difference cover
could turn out to be the most practical variant but it requires an efficient string sorter for medium length strings.

Sorting strings
using conventional \emph{atomic} sorting algorithms (that treat keys
as indivisible objects) is inefficient since comparing
entire strings can be expensive and has to be done many times in
atomic sorting algorithms. In contrast, efficient string sorting
algorithms inspect most characters of the input only once during the
entire sorting process and they inspect only those characters that are
needed to establish the global ordering. Let $D$ denote the
\emph{distinguishing prefix}, which is the minimal number of
characters that need to be inspected. Efficient sequential string
sorting algorithms come close to the lower bound of $\Om{D}$ for
sorting the input. When characters are black boxes, that can only
be compared but not further inspected, we get a lower bound of
$\Om{D+n\log n}$.  Such comparison-based string sorting algorithms
will be the main focus of our theoretical analysis.  Our
implementations also include some optimizations for integer alphabets.

Surprisingly, there has been little previous work on parallel string
sorting using $p$ processors, a.k.a. processing elements (PEs).
Here one would like to come close to time $\Oh{D/p}$ -- at least for sufficiently large inputs.
Our extensive previous work
\cite{BingSan13,bingmann2018scalable,lcpMergeSort} concentrates on
shared-memory algorithms. However, for large data sets stored on nodes
of large compute clusters, distributed-memory
algorithms are needed.
While in principle the shared-memory algorithms could be adapted, they neglect that \emph{communication volume} is the limiting factor for the scalability of algorithms to large systems \cite{amarasinghe2009exascale,Borkar13,sanders2013communication}.

The present paper largely closes this gap by developing such
communication-efficient string sorting algorithms.  After
discussing preliminaries (Section~\ref{s:prelim}) and further
related work (Section~\ref{s:related}), we begin with a very simple
baseline algorithm based on Quicksort that treats strings as atomic
objects (Section~\ref{ss:hQuick}). We then develop genuine string
sorting algorithms that are based on multi-way mergesort that was
previously used for parallel and external sorting algorithms
\cite{VarEtAl91,RahnSS10,Sundar13doi,axtmann2015practical,fischer2019lightweight}.
The data on each PE is first sorted locally.  It is then partitioned
into $p$ ranges so that one range can be moved to each PE. Finally,
each PE merges the received fragments of data.  The appeal of
multi-way merging for communication-efficient sorting is that the
local sorting exposes common prefixes of the local input strings.  Our
\emph{Distributed String Merge Sort} (\dMSS) described in
Section~\ref{s:DMSS} exploits this by only communicating the length of
the common prefix with the previous string and the remaining
characters.  The LCP values also allow us to use the multiway
LCP-merging technique we developed in \cite{lcpMergeSort} in such a
way that characters are only inspected once.  In addition, we develop
a partitioning scheme that takes the length of the strings into
account in order to achieve better load balancing.

Our second algorithm \emph{Distributed Prefix-Doubling String Merge
  Sort} (\dPDSS) described in Section~\ref{s:DPMSS} further improves
communication efficiency by only communicating characters that may be
needed to establish the global ordering of the data.  The
  algorithm also has optimal local work for a comparison-based string
  sorting algorithm.  See Theorem~\ref{thm:dPDSS} for details.
The key idea is to apply our communication-efficient duplicate
detection algorithm \cite{sanders2013communication} to geometrically
growing prefixes of each string. Once a prefix has no duplicate
anymore, we know that it is sufficient to transmit only this prefix.

In Section~\ref{s:experiments}, we present an extensive experimental
evaluation.
We observe several times better performance compared to previous approaches; in particular for large machines and strings with high potential for saving communication bandwidth.
Section~\ref{s:conclusion} concludes the paper including a
discussion of possible future work.

%%%%%%%%%%%%%%%%%%%%%%%%%%%%%%%%%%%%%%%%%%%%%%%%%%%%%%%%%%%%%%%%%%%%%%
\section{Preliminaries}\label{s:prelim}

Our input is an
array $\mathcal{S} \coloneqq [s_0, ..., s_{n-1}]$ of $n$ strings with total length $N$.
Sorting $\mathcal{S}$ amounts to permuting it so that a
lexicographical order is established.  A string $s$ of length $\ell = \abs{s}$ is
an array $[\, s[0], \ldots, s[\ell-2], 0 \,]$ where $0$ is a special end-of-string character
outside the alphabet.%
\footnote{Our algorithms can also be adapted to the case without 0-termination where the inputs specify string lengths instead.}
String arrays are usually represented
as arrays of pointers to the beginning of the strings. Thus, entire
strings can be moved or swapped in constant time.
The first $\ell$ characters of a string are its length $\ell$ \emph{prefix}.
Let $\lcpMathh{s_1}{s_2}$ denote the length of the \emph{longest common prefix} (LCP) of $s_1$ and $s_2$.
For a sorted array of strings $\StringSet{S}$, we define the corresponding \emph{LCP array} $[\bot, h_1, h_2, \dots, h_{\abs{\StringSet{S}}-1}]$ with $h_i \coloneqq \lcpMathh{s_{i-1}}{s_i}$. The string sorting algorithms we describe here produce the LCP-array as additional output. This is useful in many applications. For example, it facilitates building
a search tree that allows searching for a string pattern $s$ in time $\Oh{|s|+\log n}$ \cite{bayer1977prefix,GraLar01}.

The distinguishing prefix length $\dpre{s}$ of a string $s$ is
the number of characters that must be inspected to differentiate it
from all other strings in the set $\StringSet{S}$. We have
$\dpre{s} = \max_{t \in \StringSet{S}, t\neq s}\lcpMathh{s}{t} + 1$.  The sum of the distinguishing prefix
lengths $D$ is a lower bound on the number of characters that must be
inspected to sort the input.

Our model of computation is a distributed-memory machine with $p$ PEs.
Sending a message of $m$ bits from one PE to another PE takes time
$\alpha + \beta m$
\cite{fraigniaud1994methods,SMDD19}.\footnote{Usually, the unit is a
  different one, e.g., machine words. Here we use bits in order to be
  able to make more precise statements with respect to the number of
  characters to be communicated.}  Analyzing the communication cost of
our algorithms is mostly based on plugging in the cost of well-known
collective communication operations.  When $h$ is the maximum amount
of data sent or received at any PE, we get
$\Oh{\alpha \log  p+\beta h}$ for broadcast, reduction, and all-to-all broadcast
(a.k.a. gossiping). For personalized all-to-all communication we have a tradeoff between
low communication volume (cost $\Oh{\alpha p+\beta h}$)
and low latency (cost $\Oh{\alpha \log p+\beta h\log p}$); e.g., \cite{SMDD19}.

Table~\ref{tab:notation} summarizes the notation, concentrating on
the symbols that are needed for the result of the algorithm analysis.
%----------------------------------------------------------------------
\subsection{Sequential String Sorting for the Base Case}\label{s:seqsort}

In \cite{bingmann2018scalable} an extensive evaluation of sequential
string sorting algorithms is given in which a variant of \emph{MSD
  String Radix Sort} has been found to be among the fastest algorithms
on many data sets. We are using this algorithm for our
implementations.  This recursive algorithm considers subproblems where
all strings have a common prefix length $\ell$.  The strings are then
partitioned based on their $(\ell+1)$-st character. The recursion
stops when the subproblem contains less then $\sigma$ strings.  This
takes time $\Oh{D}$ (not counting the base case problems).  These
small subproblems are sorted using Multikey Quicksort
\cite{bentley1997fast}.  This is an adaptation of Quicksort to strings
that needs expected time $\Oh{D+n\log n}$.  Our implementation, in
turn, uses LCP insertion sort \cite{bingmann2018scalable} as a based
case for constant size inputs. This algorithm has complexity
$\Oh{D+n^2}$.  Putting these components together leads to a base case sorter with
cost $\Oh{D+n\log\sigma}$.  We have modified these algorithms so that
they produce an LCP array as part of the output at no additional cost.
The modified implementations have been made available as part of the
tlx library \cite{TLX}.

Our study \cite{bingmann2018scalable} identifies several other
efficient sequential string sorters.  Which ones are best depends on
the characteristics of the input. For example, for large alphabets and
skewed inputs strings, sample sort~\cite{BingSan13} might be better. The resulting
asymptotic complexity for such purely comparison-based algorithms is
$\Oh{D +n\log n}$ which represents a lower bound for string sorting
based on character comparisons.

\begin{table}[t]
  \caption{\label{tab:notation}Summary of Notation for the Algorithm Analysis}
  \begin{tabular}{cl}
    Symbol   & Meaning\\\hline
    $n$      & total number of input strings\\
    $N$      & total number of input characters\\
    $\sigma$ & alphabet size\\
    $D$      & total distinguishing prefix size\\
    $\nmax$  & max. number of strings on any PE\\
    $\Nmax$  & max. number of characters on any PE\\
    $\Dmax$  & max. number of distinguishing prefix characters on any PE\\
    $\lmax$  & length of the longest input string\\
    $\dmax$  & length of the longest distinguishing prefix\\\hline
    $p$      & number of processing elements (PEs)\\
    $\alpha$ & message startup latency\\
    $\beta$  & time per bit of communicated data
  \end{tabular}
\end{table}

%----------------------------------------------------------------------
\subsection{Multiway LCP-Merging}

We are using our \emph{LCP loser tree} \cite{lcpMergeSort}.  This is a
generalization of the binary merging technique proposed by Ng and
Kakehi \cite{ng2008merging} building on the (atomic) loser tree data
structure \cite{knuth1998sorting}.

A $K$-way (atomic) loser tree (a.k.a. tournament tree) is a binary
tree with $K$ leaves. Each leaf is associated with one current element
of a sorted sequence of objects -- initially the smallest element in that sequence.
This current element is passed up the tree.  Internal
nodes store the larger of the elements passed up to them (the loser)
and pass up the smaller element (the winner) to the next level. The
element passed up by the root is the globally smallest element.
This element is output in each step.
The sequence corresponding to the winner's leaf is advanced to the
next element. The data structure invariant of the loser tree can be
reestablished in logarithmic time by repairing it step by step while
going upwards from the winner's leaf to the root.
This also determines the next element to be output.

Loser trees are adapted to strings by associating each sorted
sequence with its LCP array.  Moreover, internal nodes store the intermediate LCP
length of the compared strings. The output is the sorted sequence
representing all input sequences plus the corresponding LCP array. The
number of character comparisons for multiway LCP-merging of $m$ strings
is bounded by $m\log K+\Delta L$ where $\Delta L$ is the total
increment of the LCP-array entries of the input strings. Embedded into
a string sorting algorithm this leads to total complexity $\Oh{D+n\log n}$
for sorting $n$ strings.

%----------------------------------------------------------------------
\subsection{Distributed Multiway Mergesort}\label{ss:dmm}
A starting point for our algorithms is the distributed-memory
mergesort algorithm by Fischer and Kurpicz
\cite{fischer2019lightweight} as a subroutine for suffix array
construction.  The data is first sorted locally using a sequential string sorting algorithm.
It is then partitioned globally by $p-1$ splitter strings $f_1$,\ldots, $f_{p-1}$
such that PE $i$ gets all the strings $s$ with $f_{i} < s \leq
f_{i+1}$ (with $f_0$ denoting an ``infinitely'' small string and $f_p$
an ``infinitely'' large one).
Fischer and Kurpicz choose these splitters based on a deterministic sampling technique where each PE chooses
$p-1$ samples equidistantly from its sorted local input.
After gathering the samples on PE 0,
the splitters are chosen equidistantly from the globally sorted sample.
They use an ordinary (not LCP-aware) loser tree for merging strings.

%%%%%%%%%%%%%%%%%%%%%%%%%%%%%%%%%%%%%%%%%%%%%%%%%%%%%%%%%%%%%%%%%%%%%%
\section{More Related Work}\label{s:related}

This paper is based on the master's thesis of Matthias Schimek
\cite{Schimek19}.  There has been intensive work on sequential string
sorting. Refer to
\cite{karkkainen2008engineering,sinha2010engineering,bingmann2018scalable}
for an overview of results and comparative studies.  There are very
fast PRAM algorithms for sorting strings of total length $N$ with work
$\Oh{N\log N}$ ($\Oh{N\log\log N}$ for integer characters), e.g.,
\cite{Hagerup94,jaja1994efficient}. Note that our results need
only \emph{linear} work in the (possibly much smaller)
distinguishing prefix length $D$ rather than in the total input size
$N$. The previous algorithms use a doubling technique similar to the
one used by Manber and Myers \cite{manber1993suffix} for suffix
sorting: Use integer sorting to build lexicographic names of
substrings with a length that doubles in every iteration. The
doubling technique in our \dPDSS-algorithm is much simpler -- it
only requires hashing of prefixes of the strings. Also, doubling is
not inherent in this technique but only one special case. To achieve
better approximation of distinguishing prefix lengths one can also
uses smaller multipliers.
Neelima et al. \cite{NeelimaNP14} study string sorting on GPUs.

%%%%%%%%%%%%%%%%%%%%%%%%%%%%%%%%%%%%%%%%%%%%%%%%%%%%%%%%%%%%%%%%%%%%%%
\section{Parallel String Sorting Based on\\ Atomic Parallel Quicksort}\label{ss:hQuick}

This section serves two purposes. We describe a simple parallel string
sorting algorithm whose analysis can serve as a basis for comparing it
with the more sophisticated algorithms below. We also use this
algorithm as a subroutine in the others.

This algorithm -- \hQuick -- is a rather straightforward adaptation
of an atomic sorting algorithm based on a Quicksort variant introduced
in \cite{axtmann2017robust}. We therefore only outline it, focusing on
the changes needed for string sorting. Let $d=\floor{\log p}$.  The
algorithm employs only $2^d\geq p/2$ PEs which it logically
arranges as a $d$-dimensional hypercube. The algorithm starts by
moving each input string to a random hypercube node.
\hQuick proceeds in $d$ iterations.  In iteration $i=d,\ldots,1$,
the remaining task is to sort the data within $i$-dimensional
subcubes of this hypercube.  To establish the loop invariant for the
next iteration, a pivot string $s$ is determined as a good
approximation of the median of the strings within each subcube. This
is done using a special kind of tree reduction. One subcube will then
work on the strings $\leq s$ and one works on the strings $>s$. A
tie breaking scheme enforces that the pivot is unique.  When the loop has terminated,
the remaining problem is to locally sort the data on one PE.

\begin{theorem}\label{thm:hQuick}\sloppypar
  With the notation from Table~\ref{tab:notation},
  Algorithm~\hQuick needs
  local work $\Oh{n\lmax/p\log n}$,
  latency $\Oh{\alpha\log^2 p}$,
  and bottleneck communication volume
  $\Oh{(\Nmax+n\lmax/p\log p + \lmax\log^2p)\log\sigma}$ bits in expectation.%
  \footnote{This conservative bound ensues if $n/p$
    strings of length $\lmax$ are concentrated on a single
    PE. Randomization makes this unlikely. However, the ordering of
    the strings might enforce such a distribution at the end. Hence,
    it may be possible to improve the work and communication bounds by
    a factor $\log p$.}
\end{theorem}

Before turning to the analysis, we interpret what this theorem says.
Algorithm~\hQuick is not very efficient because all the data is moved
a logarithmic number of times and because using an approximation of
the median as a pivot only balances the \emph{number} of strings but
not their total length. Also, string comparisons do not exploit
information on common prefixes that may be implicitly available. On
the other hand, the algorithm has only polylogarithmic latency which makes
it a good candidate for sorting small inputs.

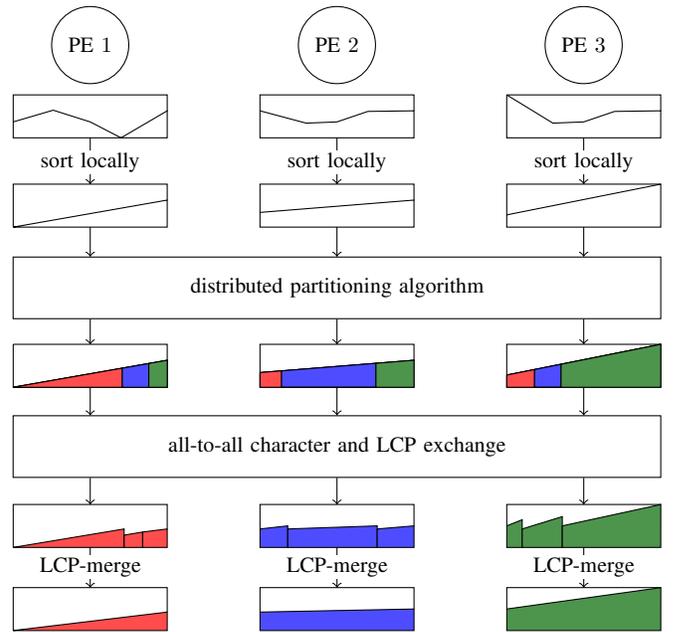
\begin{figure}
  \centering
  \tikzstyle{normalNode}=[inner sep=0.5mm, minimum size=1.25cm, circle, draw = black]
  \begin{tikzpicture}[scale=0.82,transform shape]
    \def\yup{0.65}
    \def\ydown{-0.65}
    \def\yone{0}
    \def\dist{-1}
    \def\xStart{0}
    \def\xEnd{15}

    \def\xdist{2cm}
    \def\below{0.5cm}
    \def\belowtwo{1.0cm}
    \def\belowSorted{4.25cm}
    \def\belowSorter{3.0cm}
    \def\belowCommunicator{6.0cm}
    \def\belowFinal{7.0cm}
    \def\belowAllgather{8.5cm}
    \def\belowResult{9.5cm}
    \def\belowBuckts{11.0cm}
    \def\beloww{-6.5cm}
    \def\posSorter{-4cm}
    \def\regularSize{0.45cm}
    \def\arrayHeight{0.3cm}

    \def\belowInputF{4mm}
    \def\belowInputA{12mm}
    \def\belowInput{14mm}
    \def\belowBucket{28mm}
    \def\belowLocallySorted{23.5mm}
    \def\belowLocallySortedSmaller{23.5mm}
    \def\belowStringExchange{53.75mm}
    \def\belowInputE{11mm}
    \def\inputYBottom{-2cm}
    \def\inputYTop{-2cm}
    \def\inputXLeft{0}

    \def\x{1.0}
    \def\y{0.1}

    \def\darkgreen{black!60!green!70}
    \def\red{red!70}
    \def\blue{blue!70}

    \tikzstyle{rectangleSmall}=[inner sep=0.5mm, minimum width=1cm, minimum height = 0.75cm, rectangle, draw = black]
    \tikzstyle{rectangleNode}=[inner sep=0.5mm, minimum width=2.5cm, minimum height = 0.75cm, rectangle, draw = black]
    \tikzstyle{bigRectangle}=[inner sep=0.5mm, minimum width=105mm, minimum height = 1.0cm, rectangle, draw = black]

    \draw node (PE1) at (0, 0) [normalNode] {PE $1$};
    \draw node (PE2) at (2 * \xdist, 0) [normalNode] {PE $2$};
    \draw node (PE3) at (4 * \xdist, 0) [normalNode] {PE $3$};

    \draw node [below = \belowInputF of PE1] (PE1Middle) {};
    \draw node [below = \belowInputF of PE2] (PE2Middle) {};
    \draw node [below = \belowInputF of PE3] (PE3Middle) {};

    % INPUT
    \draw node (PE11) [above left  =\y and \x of PE1Middle] {};
    \draw node (PE12) [below left  =\y and \x of PE1Middle] {};
    \draw node (PE13) [below right =\y and \x of PE1Middle] {};
    \draw node (PE14) [above right =\y and \x of PE1Middle] {};

    \draw node (PE21) [above left  =\y and \x of PE2Middle] {};
    \draw node (PE22) [below left  =\y and \x of PE2Middle] {};
    \draw node (PE23) [below right =\y and \x of PE2Middle] {};
    \draw node (PE24) [above right =\y and \x of PE2Middle] {};

    \draw node (PE31) [above left  =\y and \x of PE3Middle] {};
    \draw node (PE32) [below left  =\y and \x of PE3Middle] {};
    \draw node (PE33) [below right =\y and \x of PE3Middle] {};
    \draw node (PE34) [above right =\y and \x of PE3Middle] {};

    % path 1
    \draw node (0PE1) [above =0.1*\y                  of PE12] {};
    \draw node (1PE1) [above right =\y*2 and 0.4*\x   of PE12] {};
    \draw node (2PE1) [above right =\y*0.1 and \x     of PE12] {};
    \draw node (3PE1) [right =1.5*\x                  of PE12] {};
    \draw node (4PE1) [above =1.9*\y                  of PE13] {};

    \draw (0PE1.center) -- (PE11.center) -- (PE14.center) -- (4PE1.center);
    \draw (0PE1.center) -- (1PE1.center) -- (2PE1.center) -- (3PE1.center) -- (4PE1.center) -- (PE13.center) -- (PE12.center) -- (0PE1.center);

    % path 2
    \draw node (0PE2) [above =1.9*\y                   of PE22] {};
    \draw node (1PE2) [above right =-\y*0.1 and 0.5*\x of PE22] {};
    \draw node (2PE2) [above right =\y*0.1 and \x      of PE22] {};
    \draw node (3PE2) [above right =1.8*\y and 1.5*\x  of PE22] {};
    \draw node (4PE2) [above =1.9*\y                   of PE23] {};

    \draw (0PE2.center) -- (PE21.center) -- (PE24.center) -- (4PE2.center);
    \draw (0PE2.center) -- (1PE2.center) -- (2PE2.center) -- (3PE2.center) -- (4PE2.center) -- (PE23.center) -- (PE22.center) -- (0PE2.center);

    % path 3
    \draw node (0PE3) at (PE31) {};
    \draw node (1PE3) [above right =-\y*0.1 and 0.5*\x of PE32] {};
    \draw node (2PE3) [above right =\y*0.1 and \x      of PE32] {};
    \draw node (3PE3) [above right =1.8*\y and 1.5*\x  of PE32] {};
    \draw node (4PE3) [above =1.9*\y                   of PE33] {};

    \draw (0PE3.center) -- (PE31.center) -- (PE34.center) -- (4PE3.center);
    \draw (0PE3.center) -- (1PE3.center) -- (2PE3.center) -- (3PE3.center) -- (4PE3.center) -- (PE33.center) -- (PE32.center) -- (0PE3.center);

    %%%%%%%%%%%%%%%%%%%%%%%%%%%%%%%%%%%%%%%%%%%%%%%%%%%%%%%%%%%%%
    %
    % Local Sorting
    %
    %%%%%%%%%%%%%%%%%%%%%%%%%%%%%%%%%%%%%%%%%%%%%%%%%%%%%%%%%%%%

    \draw node [below = \belowInputA of PE1Middle](PE1MiddleLS) {};
    \draw node [below = \belowInputA of PE2Middle](PE2MiddleLS) {};
    \draw node [below = \belowInputA of PE3Middle](PE3MiddleLS) {};

    \draw node (PE11LS) [above left  =\y and \x of PE1MiddleLS] {};
    \draw node (PE12LS) [below left  =\y and \x of PE1MiddleLS] {};
    \draw node (PE13LS) [below right =\y and \x of PE1MiddleLS] {};
    \draw node (PE14LS) [above right =\y and \x of PE1MiddleLS] {};

    \draw node (PE21LS) [above left  =\y and \x of PE2MiddleLS] {};
    \draw node (PE22LS) [below left  =\y and \x of PE2MiddleLS] {};
    \draw node (PE23LS) [below right =\y and \x of PE2MiddleLS] {};
    \draw node (PE24LS) [above right =\y and \x of PE2MiddleLS] {};

    \draw node (PE31LS) [above left  =\y and \x of PE3MiddleLS] {};
    \draw node (PE32LS) [below left  =\y and \x of PE3MiddleLS] {};
    \draw node (PE33LS) [below right =\y and \x of PE3MiddleLS] {};
    \draw node (PE34LS) [above right =\y and \x of PE3MiddleLS] {};

    % path 1
    \draw node (0PE1LS) at(PE12LS) {};
    \draw node (1PE1LS) [above =1.9*\y                  of PE13LS] {};

    \draw (0PE1LS.center) -- (PE11LS.center) -- (PE14LS.center) -- (1PE1LS.center);
    \draw (0PE1LS.center) -- (1PE1LS.center) -- (PE13LS.center) -- (PE12LS.center) -- (0PE1LS.center);

    % path 2
    \draw node (0PE2LS) [above =-0.1*\y                   of PE22LS] {};
    \draw node (1PE2LS) [above =1.9*\y                   of PE23LS] {};

    \draw (0PE2LS.center) -- (PE21LS.center) -- (PE24LS.center) -- (1PE2LS.center);
    \draw (0PE2LS.center) -- (1PE2LS.center) -- (PE23LS.center) -- (PE22LS.center) -- (0PE2LS.center);

    % path 3
    \draw node (0PE3LS) [above =-0.5*\y                  of PE32LS] {};
    \draw node (1PE3LS) at (PE34LS) {};

    \draw (0PE3LS.center) -- (PE31LS.center) -- (PE34LS.center) -- (1PE3LS.center);
    \draw (0PE3LS.center) -- (1PE3LS.center) -- (PE33LS.center) -- (PE32LS.center) -- (0PE3LS.center);

    \draw node (partitioner) [bigRectangle, below = \belowBucket of PE2]  {distributed partitioning algorithm};

    %%%%%%%%%%%%%%%%%%%%%%%%%%%%%%%%%%%%%%%%%%%%%%%%%%%%%%%%%%%%%
    %
    % Partitioned
    %
    %%%%%%%%%%%%%%%%%%%%%%%%%%%%%%%%%%%%%%%%%%%%%%%%%%%%%%%%%%%%

    \draw node [below = \belowLocallySorted of PE1MiddleLS](PE1MiddlePart) {};
    \draw node [below = \belowLocallySorted of PE2MiddleLS](PE2MiddlePart) {};
    \draw node [below = \belowLocallySorted of PE3MiddleLS](PE3MiddlePart) {};

    \draw node (PE11Part) [above left  =\y and \x of PE1MiddlePart] {};
    \draw node (PE12Part) [below left  =\y and \x of PE1MiddlePart] {};
    \draw node (PE13Part) [below right =\y and \x of PE1MiddlePart] {};
    \draw node (PE14Part) [above right =\y and \x of PE1MiddlePart] {};

    \draw node (PE21Part) [above left  =\y and \x of PE2MiddlePart] {};
    \draw node (PE22Part) [below left  =\y and \x of PE2MiddlePart] {};
    \draw node (PE23Part) [below right =\y and \x of PE2MiddlePart] {};
    \draw node (PE24Part) [above right =\y and \x of PE2MiddlePart] {};

    \draw node (PE31Part) [above left  =\y and \x of PE3MiddlePart] {};
    \draw node (PE32Part) [below left  =\y and \x of PE3MiddlePart] {};
    \draw node (PE33Part) [below right =\y and \x of PE3MiddlePart] {};
    \draw node (PE34Part) [above right =\y and \x of PE3MiddlePart] {};

    % path 1
    \draw node (0PE1Part) at(PE12Part) {};
    \draw node (1PE1Part) [above =1.9*\y                  of PE13Part] {};
    \draw let \p1=(PE11Part) in node (0PE1PartUP) at(0.95cm, \y1) {};
    \draw let \p1=(PE12Part) in node (0PE1PartDOWN) at(0.95cm, \y1) {};
    \draw let \p1=(PE11Part) in node (1PE1PartUP) at(0.52cm, \y1) {};
    \draw let \p1=(PE12Part) in node (1PE1PartDOWN) at(0.52cm, \y1) {};

    \draw (PE11Part.center) -- (PE12Part.center) -- (PE13Part.center) -- (PE14Part.center) -- (PE11Part.center);
    \draw (0PE1Part.center) -- (1PE1Part.center) -- (PE13Part.center) -- (PE12Part.center) -- (0PE1Part.center);
    \path [name path=path1](0PE1PartUP.center) -- (0PE1PartDOWN.center);
    \path [name path=path2](1PE1PartUP.center) -- (1PE1PartDOWN.center);
    \draw [name path=diag](0PE1Part.center) -- (1PE1Part.center);
    \path [name intersections={of=path2 and diag,by=E}];
    \path [name intersections={of=path1 and diag,by=F}];

    \draw [fill=\red](0PE1Part.center) -- (1PE1PartDOWN.center) -- (E.center) -- (0PE1Part.center);
    \draw [fill=\blue](1PE1PartDOWN.center) -- (E.center) -- (F.center) -- (0PE1PartDOWN.center) -- (1PE1PartDOWN.center);
    \draw [fill=\darkgreen](0PE1PartDOWN.center) -- (F.center) -- (1PE1Part.center) -- (PE13Part.center) -- (0PE1PartDOWN.center);

    % path 2
    \draw node (0PE2Part) [above =-0.1*\y                   of PE22Part] {};
    \draw node (1PE2Part) [above =1.9*\y                    of PE23Part] {};
    \draw let \p1=(PE21Part) in node (0PE2PartUP)   at(\x1 + 0.35cm, \y1) {};
    \draw let \p1=(PE22Part) in node (0PE2PartDOWN) at(\x1 + 0.35cm, \y1) {};
    \draw let \p1=(PE21Part) in node (1PE2PartUP)   at(\x1 + 1.88cm, \y1) {};
    \draw let \p1=(PE22Part) in node (1PE2PartDOWN) at(\x1 + 1.88cm, \y1) {};

    \draw (0PE2Part.center) -- (PE21Part.center) -- (PE24Part.center) -- (1PE2Part.center);
    \draw (0PE2Part.center) -- (1PE2Part.center) -- (PE23Part.center) -- (PE22Part.center) -- (0PE2Part.center);

    \path [name path=path1](0PE2PartUP.center) -- (0PE2PartDOWN.center);
    \path [name path=path2](1PE2PartUP.center) -- (1PE2PartDOWN.center);
    \draw [name path=diag](0PE2Part.center) -- (1PE2Part.center);
    \path [name intersections={of=path1 and diag,by=E}];
    \path [name intersections={of=path2 and diag,by=F}];
    \draw [fill=\red](0PE2Part.center) -- (E.center) -- (0PE2PartDOWN.center) -- (PE22Part.center) -- (0PE2Part.center);
    \draw [fill=\blue](0PE2PartDOWN.center) -- (E.center) -- (F.center) -- (1PE2PartDOWN.center) -- (0PE2PartDOWN.center);
    \draw [fill=\darkgreen](1PE2PartDOWN.center) -- (F.center) -- (1PE2Part.center) -- (PE23Part.center) -- (0PE2PartDOWN.center);
    %
    % path 3
    \draw node (0PE3Part) [above =-0.5*\y                  of PE32Part] {};
    \draw node (1PE3Part) at (PE34Part) {};
    \draw let \p1=(PE31Part) in node (0PE3PartUP)   at(\x1 + 0.45cm, \y1) {};
    \draw let \p1=(PE32Part) in node (0PE3PartDOWN) at(\x1 + 0.45cm, \y1) {};
    \draw let \p1=(PE31Part) in node (1PE3PartUP)   at(\x1 + 0.88cm, \y1) {};
    \draw let \p1=(PE32Part) in node (1PE3PartDOWN) at(\x1 + 0.88cm, \y1) {};

    \draw (0PE3Part.center) -- (PE31Part.center) -- (PE34Part.center) -- (1PE3Part.center);
    \draw (0PE3Part.center) -- (1PE3Part.center) -- (PE33Part.center) -- (PE32Part.center) -- (0PE3Part.center);

    \path [name path=path1](0PE3PartUP.center) -- (0PE3PartDOWN.center);
    \path [name path=path2](1PE3PartUP.center) -- (1PE3PartDOWN.center);
    \draw [name path=diag](0PE3Part.center) -- (1PE3Part.center);
    \path [name intersections={of=path1 and diag,by=E}];
    \path [name intersections={of=path2 and diag,by=F}];
    \draw [fill=\red](0PE3Part.center) -- (E.center) -- (0PE3PartDOWN.center) -- (PE32Part.center) -- (0PE3Part.center);
    \draw [fill=\blue](0PE3PartDOWN.center) -- (E.center) -- (F.center) -- (1PE3PartDOWN.center) -- (0PE3PartDOWN.center);
    \draw [fill=\darkgreen](1PE3PartDOWN.center) -- (F.center) -- (1PE3Part.center) -- (PE33Part.center) -- (0PE3PartDOWN.center);

    \draw node (buckets) [bigRectangle, below = \belowStringExchange of PE2]  {all-to-all character and LCP exchange};

    % arrows
    \draw [->] (PE1MiddlePart|-PE12Part) -- (PE1MiddlePart|-buckets.north);
    \draw [->] (PE2MiddlePart|-PE22Part) -- (PE2MiddlePart|-buckets.north);
    \draw [->] (PE3MiddlePart|-PE32Part) -- (PE3MiddlePart|-buckets.north);

    \draw [<-] (PE1MiddlePart|-PE11Part) -- (PE1MiddlePart|-partitioner.south);
    \draw [<-] (PE2MiddlePart|-PE21Part) -- (PE2MiddlePart|-partitioner.south);
    \draw [<-] (PE3MiddlePart|-PE31Part) -- (PE3MiddlePart|-partitioner.south);

    \draw [->] (PE1MiddleLS|-PE12LS) -- (PE1MiddlePart|-partitioner.north);
    \draw [->] (PE2MiddleLS|-PE22LS) -- (PE2MiddlePart|-partitioner.north);
    \draw [->] (PE3MiddleLS|-PE32LS) -- (PE3MiddlePart|-partitioner.north);

    \draw [<-] (PE1MiddleLS|-PE11LS) -- (PE1Middle|-PE12) node[above, inner sep=1pt, pos=0.2, fill=white] {sort locally};
    \draw [<-] (PE2MiddleLS|-PE21LS) -- (PE2Middle|-PE22) node[above, inner sep=1pt, pos=0.2, fill=white] {sort locally};
    \draw [<-] (PE3MiddleLS|-PE31LS) -- (PE3Middle|-PE32) node[above, inner sep=1pt, pos=0.2, fill=white] {sort locally};

    %%%%%%%%%%%%%%%%%%%%%%%%%%%%%%%%%%%%%%%%%%%%%%%%%%%%%%%%%%%%%
    %
    % Exchanged
    %
    %%%%%%%%%%%%%%%%%%%%%%%%%%%%%%%%%%%%%%%%%%%%%%%%%%%%%%%%%%%%

    \draw node [below = \belowLocallySortedSmaller of PE1MiddlePart](PE1MiddleEx) {};
    \draw node [below = \belowLocallySortedSmaller of PE2MiddlePart](PE2MiddleEx) {};
    \draw node [below = \belowLocallySortedSmaller of PE3MiddlePart](PE3MiddleEx) {};

    \draw node (PE11Ex) [above left  =\y and \x of PE1MiddleEx] {};
    \draw node (PE12Ex) [below left  =\y and \x of PE1MiddleEx] {};
    \draw node (PE13Ex) [below right =\y and \x of PE1MiddleEx] {};
    \draw node (PE14Ex) [above right =\y and \x of PE1MiddleEx] {};

    \draw node (PE21Ex) [above left  =\y and \x of PE2MiddleEx] {};
    \draw node (PE22Ex) [below left  =\y and \x of PE2MiddleEx] {};
    \draw node (PE23Ex) [below right =\y and \x of PE2MiddleEx] {};
    \draw node (PE24Ex) [above right =\y and \x of PE2MiddleEx] {};

    \draw node (PE31Ex) [above left  =\y and \x of PE3MiddleEx] {};
    \draw node (PE32Ex) [below left  =\y and \x of PE3MiddleEx] {};
    \draw node (PE33Ex) [below right =\y and \x of PE3MiddleEx] {};
    \draw node (PE34Ex) [above right =\y and \x of PE3MiddleEx] {};

    \draw [<-] (PE1MiddleEx|-PE11Ex) -- (PE1MiddleEx|-buckets.south);
    \draw [<-] (PE2MiddleEx|-PE21Ex) -- (PE2MiddleEx|-buckets.south);
    \draw [<-] (PE3MiddleEx|-PE31Ex) -- (PE3MiddleEx|-buckets.south);

    % path1
    %%%%%%%%%%%%%%%%%%%%%%%%%%%%%%%%%%%%%%%%%%%%%%%%%%%%%%%%%%%%%
    %
    % Exchanged
    %
    %%%%%%%%%%%%%%%%%%%%%%%%%%%%%%%%%%%%%%%%%%%%%%%%%%%%%%%%%%%%

    \draw node [below = \belowLocallySorted of PE1MiddlePart](PE1MiddleEx) {};
    \draw node [below = \belowLocallySorted of PE2MiddlePart](PE2MiddleEx) {};
    \draw node [below = \belowLocallySorted of PE3MiddlePart](PE3MiddleEx) {};

    \draw node (PE11Ex) [above left  =\y and \x of PE1MiddleEx] {};
    \draw node (PE12Ex) [below left  =\y and \x of PE1MiddleEx] {};
    \draw node (PE13Ex) [below right =\y and \x of PE1MiddleEx] {};
    \draw node (PE14Ex) [above right =\y and \x of PE1MiddleEx] {};

    \draw node (PE21Ex) [above left  =\y and \x of PE2MiddleEx] {};
    \draw node (PE22Ex) [below left  =\y and \x of PE2MiddleEx] {};
    \draw node (PE23Ex) [below right =\y and \x of PE2MiddleEx] {};
    \draw node (PE24Ex) [above right =\y and \x of PE2MiddleEx] {};

    \draw node (PE31Ex) [above left  =\y and \x of PE3MiddleEx] {};
    \draw node (PE32Ex) [below left  =\y and \x of PE3MiddleEx] {};
    \draw node (PE33Ex) [below right =\y and \x of PE3MiddleEx] {};
    \draw node (PE34Ex) [above right =\y and \x of PE3MiddleEx] {};

    % path1
    \def\xeins{0.55cm}
    \def\xzwei{0.85cm}
    \draw let \p1=(PE12Ex) in node (1PE1Ex) at(\xeins, \y1 + 0.3cm) {};
    \draw let \p1=(PE12Ex) in node (2PE1Ex) at(\xeins, \y1) {};
    \draw let \p1=(PE12Ex) in node (3PE1Ex) at(\xeins, \y1 + 0.2cm) {};
    \draw let \p1=(PE12Ex) in node (4PE1Ex) at(\xzwei, \y1 + 0.25cm) {};
    \draw let \p1=(PE12Ex) in node (5PE1Ex) at(\xzwei, \y1) {};
    \draw let \p1=(PE12Ex) in node (6PE1Ex) at(\xzwei, \y1 + 0.25cm) {};
    \draw let \p1=(PE13Ex) in node (7PE1Ex) at(\x1,  \y1 + 0.3cm) {};
    \draw (PE12Ex.center) -- (PE11Ex.center) -- (PE14Ex.center) -- (7PE1Ex.center);
    \draw [fill = \red](PE12Ex.center) -- (1PE1Ex.center) -- (2PE1Ex.center) -- (PE12Ex.center);
    \draw [fill = \red](3PE1Ex.center) -- (4PE1Ex.center) -- (5PE1Ex.center) -- (2PE1Ex.center) -- (3PE1Ex.center);
    \draw [fill = \red](6PE1Ex.center) -- (7PE1Ex.center) -- (PE13Ex.center) -- (5PE1Ex.center) -- (6PE1Ex.center);

    % path2
    \def\xeins{0.45cm}
    \def\xzwei{1.90cm}
    \draw let \p1=(PE22Ex) in node (1PE2Ex) at(\x1, \y1 + 0.3cm) {};
    \draw let \p1=(PE22Ex) in node (2PE2Ex) at(\x1 + \xeins, \y1 + 0.35cm) {};
    \draw let \p1=(PE22Ex) in node (3PE2Ex) at(\x1 + \xeins, \y1) {};
    \draw let \p1=(PE22Ex) in node (4PE2Ex) at(\x1 + \xeins, \y1 + 0.30cm) {};
    \draw let \p1=(PE22Ex) in node (5PE2Ex) at(\x1 + \xzwei, \y1 +0.35cm) {};
    \draw let \p1=(PE22Ex) in node (6PE2Ex) at(\x1 + \xzwei, \y1) {};
    \draw let \p1=(PE22Ex) in node (7PE2Ex) at(\x1 + \xzwei, \y1 +0.30cm) {};
    \draw let \p1=(PE23Ex) in node (8PE2Ex) at(\x1,  \y1 + 0.35cm) {};
    \draw (1PE2Ex.center) -- (PE21Ex.center) -- (PE24Ex.center) -- (8PE2Ex.center);
    \draw [fill = \blue](1PE2Ex.center) -- (2PE2Ex.center) -- (3PE2Ex.center) -- (PE22Ex.center) -- (1PE2Ex.center);
    \draw [fill = \blue](4PE2Ex.center) -- (5PE2Ex.center) -- (6PE2Ex.center) -- (3PE2Ex.center) -- (4PE2Ex.center);
    \draw [fill = \blue](7PE2Ex.center) -- (8PE2Ex.center) -- (PE23Ex.center) -- (6PE2Ex.center) -- (7PE2Ex.center);

    % path3
    \def\xeins{0.25cm}
    \def\xzwei{0.90cm}
    \draw let \p1=(PE32Ex) in node (1PE3Ex) at(\x1, \y1 + 0.35cm) {};
    \draw let \p1=(PE32Ex) in node (2PE3Ex) at(\x1 + \xeins, \y1 + 0.45cm) {};
    \draw let \p1=(PE32Ex) in node (3PE3Ex) at(\x1 + \xeins, \y1) {};
    \draw let \p1=(PE32Ex) in node (4PE3Ex) at(\x1 + \xeins, \y1 + 0.30cm) {};
    \draw let \p1=(PE32Ex) in node (5PE3Ex) at(\x1 + \xzwei, \y1 +0.5cm) {};
    \draw let \p1=(PE32Ex) in node (6PE3Ex) at(\x1 + \xzwei, \y1) {};
    \draw let \p1=(PE32Ex) in node (7PE3Ex) at(\x1 + \xzwei, \y1 +0.35cm) {};
    \draw let \p1=(PE34Ex) in node (8PE3Ex) at(\x1,  \y1) {};
    \draw (1PE3Ex.center) -- (PE31Ex.center) -- (PE34Ex.center) -- (8PE3Ex.center);
    \draw [fill = \darkgreen](1PE3Ex.center) -- (2PE3Ex.center) -- (3PE3Ex.center) -- (PE32Ex.center) -- (1PE3Ex.center);
    \draw [fill = \darkgreen](4PE3Ex.center) -- (5PE3Ex.center) -- (6PE3Ex.center) -- (3PE3Ex.center) -- (4PE3Ex.center);
    \draw [fill = \darkgreen](7PE3Ex.center) -- (8PE3Ex.center) -- (PE33Ex.center) -- (6PE3Ex.center) -- (7PE3Ex.center);

    %%%%%%%%%%%%%%%%%%%%%%%%%%%%%%%%%%%%%%%%%%%%%%%%%%%%%%%%%%%%%
    %
    % Merged
    %
    %%%%%%%%%%%%%%%%%%%%%%%%%%%%%%%%%%%%%%%%%%%%%%%%%%%%%%%%%%%%

    \draw node [below = \belowInputE of PE1MiddleEx](PE1MiddleMerged) {};
    \draw node [below = \belowInputE of PE2MiddleEx](PE2MiddleMerged) {};
    \draw node [below = \belowInputE of PE3MiddleEx](PE3MiddleMerged) {};

    \draw node (PE11Merged) [above left  =\y and \x of PE1MiddleMerged] {};
    \draw node (PE12Merged) [below left  =\y and \x of PE1MiddleMerged] {};
    \draw node (PE13Merged) [below right =\y and \x of PE1MiddleMerged] {};
    \draw node (PE14Merged) [above right =\y and \x of PE1MiddleMerged] {};

    \draw node (PE21Merged) [above left  =\y and \x of PE2MiddleMerged] {};
    \draw node (PE22Merged) [below left  =\y and \x of PE2MiddleMerged] {};
    \draw node (PE23Merged) [below right =\y and \x of PE2MiddleMerged] {};
    \draw node (PE24Merged) [above right =\y and \x of PE2MiddleMerged] {};

    \draw node (PE31Merged) [above left  =\y and \x of PE3MiddleMerged] {};
    \draw node (PE32Merged) [below left  =\y and \x of PE3MiddleMerged] {};
    \draw node (PE33Merged) [below right =\y and \x of PE3MiddleMerged] {};
    \draw node (PE34Merged) [above right =\y and \x of PE3MiddleMerged] {};

    % merging
    \draw [<-] (PE1MiddleMerged|-PE11Merged) -- (PE1MiddleEx|-PE12Ex) node[above, inner sep=1pt, pos=0.2, fill=white] {LCP-merge};
    \draw [<-] (PE2MiddleMerged|-PE21Merged) -- (PE2MiddleEx|-PE22Ex) node[above, inner sep=1pt, pos=0.2, fill=white] {LCP-merge};
    \draw [<-] (PE3MiddleMerged|-PE31Merged) -- (PE3MiddleEx|-PE32Ex) node[above, inner sep=1pt, pos=0.2, fill=white] {LCP-merge};

    % path1
    \draw let \p1=(PE13Merged) in node (1PE1Merged) at(\x1, \y1 + 0.3cm) {};
    \draw (PE12Merged.center) -- (PE11Merged.center) -- (PE14Merged.center) -- (1PE1Merged.center);
    \draw [fill = \red](PE12Merged.center) -- (1PE1Merged.center) -- (PE13Merged.center) -- (PE12Merged.center);

    % path2
    \draw let \p1=(PE22Merged) in node (1PE2Merged) at(\x1, \y1 + 0.3cm) {};
    \draw let \p1=(PE23Merged) in node (2PE2Merged) at(\x1, \y1 + 0.35cm) {};
    \draw (1PE2Merged.center) -- (PE21Merged.center) -- (PE24Merged.center) -- (2PE2Merged.center);
    \draw [fill = \blue](1PE2Merged.center) -- (2PE2Merged.center) -- (PE23Merged.center) -- (PE22Merged.center) -- (1PE2Merged.center);

    % path3
    \draw let \p1=(PE32Merged) in node (1PE3Merged) at(\x1, \y1 + 0.35cm) {};
    \draw let \p1=(PE34Merged) in node (2PE3Merged) at(\x1, \y1) {};
    \draw (1PE3Merged.center) -- (PE31Merged.center) -- (PE34Merged.center) -- (2PE3Merged.center);
    \draw [fill = \darkgreen](1PE3Merged.center) -- (2PE3Merged.center) -- (PE33Merged.center) -- (PE32Merged.center) -- (1PE3Merged.center);

  \end{tikzpicture}
  \caption{Standard distributed mergesort scheme, which we augmented in every step with string-specific optimizations.}\label{fig:mergesort scheme}
\end{figure}

\begin{proof}
The term $\Nmax\log\sigma$ is due to the initial random placement of the input.
We assume here that afterwards the pivot selection ensures (in expectation) that
in each iteration, each PE works on $\Oh{n/p}$ strings in
expectation. See \cite{axtmann2017robust} for the details which
transfer from the atomic algorithm. We make the conservative
assumption that each string incurs work and communication volume
$\lmax$ in each iteration. Similarly, we assume that local sorting
takes time $\Oh{n\lmax/p\log n}$. The term $\lmax\log\sigma\log^2p$ in
the communication volume stems from the reduction operation that, in
each iteration, needs to transmit up to $\lmax$ characters along a
reduction tree of logarithmic depth.
\end{proof}

%%%%%%%%%%%%%%%%%%%%%%%%%%%%%%%%%%%%%%%%%%%%%%%%%%%%%%%%%%%%%%%%%%%%%%
\section{Distributed String Merge Sort}\label{s:DMSS}

Algorithm~\dMSS is based on the standard mergesort scheme (see Figure~\ref{fig:mergesort scheme}) for
distributed memory but we need to augment it in every step with
string-specific optimizations.  Each PE $i$ starts with a string array
$\mathcal{S}_i$ as input and the goal is to sort the union
$\StringSet{S}$ of all inputs such that afterwards strings of PE $i$
are larger than those on PE $i-1$, smaller than those on PE $i+1$, and locally sorted.
We also output the LCP array.  The \dMSS algorithm follows
the following four steps (see Fig.~\ref{fig:algo-MS} for an illustration):
\begin{enumerate}
\item Sort the string set locally using a sequential string sorting algorithm which also saves the local LCP array (see Section~\ref{s:seqsort} for details).
\item Determine $p-1$ global splitters $f_1,\ldots,f_{p-1}$ such that PE $i$ gets bucket $b_i$ containing all strings $s$ with $f_{i} < s \leq f_{i+1}$ assuming sentinels $f_0 = -\infty$ and $f_p = +\infty$.
\item Perform an all-to-all exchange of string and LCP data, optionally applying LCP compression.
\item Merge the $p$ received sorted subsequences locally with our efficient LCP-aware loser tree.
\end{enumerate}

The following subsections discuss details of these steps.

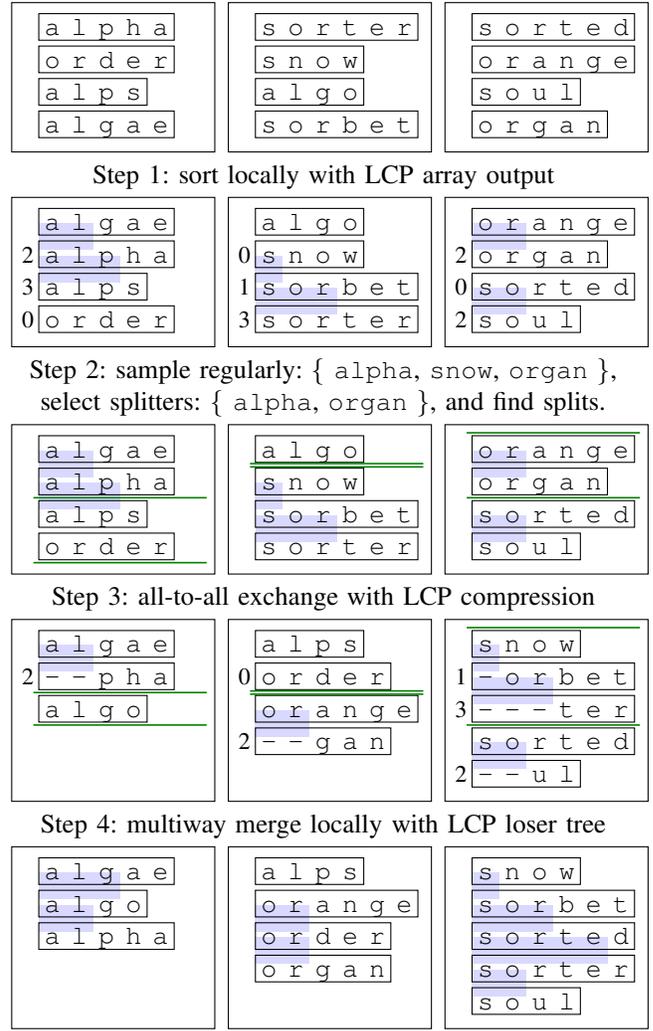
\begin{figure}
  \centering
  \begin{tikzpicture}[
    scale=0.36, yscale=1.2,
    letter/.style={
      font=\ttfamily, anchor=base,
      black, inner sep=0,
      minimum width=3.6mm, text height=2.9mm, text depth=0.7mm},
    lcpv/.style={
      anchor=base, inner sep=0,
      font={\small},
    },
    ]

    % \drawstring{y}{lcp}{s,t,r,i,n,g}
    \def\drawstring#1#2#3{
      \ifx&#2&\else
        \begin{scope}[on background layer]
          \fill[blue!15] (-0.5,-#1+1.2) rectangle (#2-0.5,-#1+0.4);
        \end{scope}
      \fi
      \foreach \c [count=\x from 0] in {#3} {
        \ifx&#2&
          \node[letter] (L#1-\x) at (\x,-#1) {\c};
        \else
        \ifthenelse{\x<#2}{
          \node[letter] (L#1-\x) at (\x,-#1) {\c};
        }{
          \node[letter] (L#1-\x) at (\x,-#1) {\c};
        }
        \fi
        \xdef\xmax{\x}
      }
      \draw (-0.5,-#1-0.16) rectangle (\xmax+0.5,-#1+0.66);
      \node[lcpv] at (-0.9,-#1) {#2};
    }

    \begin{scope}
      \begin{scope}
        \drawstring{0}{}{a,l,p,h,a};
        \drawstring{1}{}{o,r,d,e,r};
        \drawstring{2}{}{a,l,p,s};
        \drawstring{3}{}{a,l,g,a,e};
        \draw (-1.5,1) rectangle (6,-3.6);
      \end{scope}
      \begin{scope}[xshift=8cm]
        \drawstring{0}{}{s,o,r,t,e,r};
        \drawstring{1}{}{s,n,o,w};
        \drawstring{2}{}{a,l,g,o};
        \drawstring{3}{}{s,o,r,b,e,t};
        \draw (-1.5,1) rectangle (6,-3.6);
      \end{scope}
      \begin{scope}[xshift=16cm]
        \drawstring{0}{}{s,o,r,t,e,d};
        \drawstring{1}{}{o,r,a,n,g,e};
        \drawstring{2}{}{s,o,u,l};
        \drawstring{3}{}{o,r,g,a,n};
        \draw (-1.5,1) rectangle (6,-3.6);
      \end{scope}

      \node at (10,-4.55) [anchor=base] {Step 1: sort locally with LCP array output};
      \coordinate (A) at (0, -4.6);
    \end{scope}

    \begin{scope}[shift=(A),yshift=-14mm]
      \begin{scope}
        \drawstring{0}{}{a,l,g,a,e};
        \drawstring{1}{2}{a,l,p,h,a};
        \drawstring{2}{3}{a,l,p,s};
        \drawstring{3}{0}{o,r,d,e,r};
        \draw (-1.5,1) rectangle (6,-3.6);
      \end{scope}
      \begin{scope}[xshift=8cm]
        \drawstring{0}{}{a,l,g,o};
        \drawstring{1}{0}{s,n,o,w};
        \drawstring{2}{1}{s,o,r,b,e,t};
        \drawstring{3}{3}{s,o,r,t,e,r};
        \draw (-1.5,1) rectangle (6,-3.6);
      \end{scope}
      \begin{scope}[xshift=16cm]
        \drawstring{0}{}{o,r,a,n,g,e};
        \drawstring{1}{2}{o,r,g,a,n};
        \drawstring{2}{0}{s,o,r,t,e,d};
        \drawstring{3}{2}{s,o,u,l};
        \draw (-1.5,1) rectangle (6,-3.6);
      \end{scope}

      \node at (10,-4.55) [anchor=base] {Step 2: sample regularly: \{ \texttt{alpha}, \texttt{snow}, \texttt{organ} \}, };
      \node at (10,-5.55) [anchor=base] {select splitters: \{ \texttt{alpha}, \texttt{organ} \}, and find splits.};
      \coordinate (A) at (0, -4.6);
    \end{scope}

    \begin{scope}[shift=(A),yshift=-24mm]
      \tikzset{lcpv/.append style={white}}
      \begin{scope}
        \drawstring{0}{}{a,l,g,a,e};
        \drawstring{1}{2}{a,l,p,h,a};
        \drawstring{2}{3}{a,l,p,s};
        \drawstring{3}{0}{o,r,d,e,r};
        \draw (-1.5,1) rectangle (6,-3.6);
        \draw [semithick,green!50!black] (-0.7,-1.25) -- +(6.4,0);
        \draw [semithick,green!50!black] (-0.7,-3.25) -- +(6.4,0);
      \end{scope}
      \begin{scope}[xshift=8cm]
        \drawstring{0}{}{a,l,g,o};
        \drawstring{1}{0}{s,n,o,w};
        \drawstring{2}{1}{s,o,r,b,e,t};
        \drawstring{3}{3}{s,o,r,t,e,r};
        \draw (-1.5,1) rectangle (6,-3.6);
        \draw [semithick,green!50!black] (-0.7,-0.2) -- +(6.4,0);
        \draw [semithick,green!50!black] (-0.7,-0.3) -- +(6.4,0);
      \end{scope}
      \begin{scope}[xshift=16cm]
        \drawstring{0}{}{o,r,a,n,g,e};
        \drawstring{1}{2}{o,r,g,a,n};
        \drawstring{2}{0}{s,o,r,t,e,d};
        \drawstring{3}{2}{s,o,u,l};
        \draw (-1.5,1) rectangle (6,-3.6);
        \draw [semithick,green!50!black] (-0.7,0.75) -- +(6.4,0);
        \draw [semithick,green!50!black] (-0.7,-1.25) -- +(6.4,0);
      \end{scope}

      \node at (10,-4.55) [anchor=base] {Step 3: all-to-all exchange with LCP compression};
      \coordinate (A) at (0, -4.6);
    \end{scope}

    \begin{scope}[shift=(A),yshift=-14mm]
      %\tikzset{lcpv/.append style={white}}
      \begin{scope}
        \drawstring{0}{}{a,l,g,a,e};
        \drawstring{1}{2}{-,-,p,h,a};
        \drawstring{2}{}{a,l,g,o};
        \draw (-1.5,1) rectangle (6,-4.6);
        \draw [semithick,green!50!black] (-0.7,-1.25) -- +(6.4,0);
        \draw [semithick,green!50!black] (-0.7,-2.25) -- +(6.4,0);
      \end{scope}
      \begin{scope}[xshift=8cm]
        \drawstring{0}{}{a,l,p,s};
        \drawstring{1}{0}{o,r,d,e,r};
        \drawstring{2}{}{o,r,a,n,g,e};
        \drawstring{3}{2}{-,-,g,a,n};
        \draw (-1.5,1) rectangle (6,-4.6);
        \draw [semithick,green!50!black] (-0.7,-1.2) -- +(6.4,0);
        \draw [semithick,green!50!black] (-0.7,-1.3) -- +(6.4,0);
      \end{scope}
      \begin{scope}[xshift=16cm]
        \drawstring{0}{}{s,n,o,w};
        \drawstring{1}{1}{-,o,r,b,e,t};
        \drawstring{2}{3}{-,-,-,t,e,r};
        \drawstring{3}{}{s,o,r,t,e,d};
        \drawstring{4}{2}{-,-,u,l};
        \draw (-1.5,1) rectangle (6,-4.6);
        \draw [semithick,green!50!black] (-0.7,0.75) -- +(6.4,0);
        \draw [semithick,green!50!black] (-0.7,-2.25) -- +(6.4,0);
      \end{scope}

      \node at (10,-5.55) [anchor=base] {Step 4: multiway merge locally with LCP loser tree};
      \coordinate (A) at (0, -5.6);
    \end{scope}

    \begin{scope}[shift=(A),yshift=-14mm]
      \tikzset{lcpv/.append style={white}}
      \begin{scope}
        \drawstring{0}{}{a,l,g,a,e};
        \drawstring{1}{3}{a,l,g,o};
        \drawstring{2}{2}{a,l,p,h,a};
        \draw (-1.5,1) rectangle (6,-4.6);
      \end{scope}
      \begin{scope}[xshift=8cm]
        \drawstring{0}{}{a,l,p,s};
        \drawstring{1}{0}{o,r,a,n,g,e};
        \drawstring{2}{2}{o,r,d,e,r};
        \drawstring{3}{2}{o,r,g,a,n};
        \draw (-1.5,1) rectangle (6,-4.6);
      \end{scope}
      \begin{scope}[xshift=16cm]
        \drawstring{0}{}{s,n,o,w};
        \drawstring{1}{1}{s,o,r,b,e,t};
        \drawstring{2}{3}{s,o,r,t,e,d};
        \drawstring{3}{5}{s,o,r,t,e,r};
        \drawstring{4}{2}{s,o,u,l};
        \draw (-1.5,1) rectangle (6,-4.6);
      \end{scope}
    \end{scope}

  \end{tikzpicture}
  \caption{\label{fig:algo-MS}Steps of Algorithm \dMSS shown on example strings.
    The small numbers and shaded blue area after Step 1 are the calculated
    LCPs. The green lines after Step 2 are the two splitting positions. In Step
    3, characters shown as ``\texttt{-}'' are omitted due to LCP compression.}
\end{figure}

\subsection{String-Based or Character-Based Partitioning}

When determining the splitters $f_i$ in step 2) the goal is to balance the result among all PEs.
In the case of strings, this can mean to balance the number of strings or the number of characters each PE receives.
For \dMSS we thus devised a \emph{string}-based and an alternative \emph{character}-based partitioning step to determine splitters.
Both assume a given oversampling parameter $v$.
Furthermore, because we sort the local string sets $\StringSet{S}_i$ in step 1), we can use \emph{regular} sampling~\cite{shi1992parallel,li1993versatility} instead of randomly selecting samples.

String-based partitioning performs the following steps:
\begin{enumerate}
\item Each PE $i$ chooses $v$ evenly spaced samples $\StringSet{V}_i$ from its strings $\StringSet{S}_i$.
  Assuming $|\StringSet{S}_i|$ is divisible by $v + 1$, then we can choose the strings $\StringSet{S}_i[\omega j - 1]$ with $\omega = |\StringSet{S}_i|/(v+1)$ for $j = 1,\ldots,v$.
\item The $pv$ samples are globally sorted into array $\StringSet{V}$. Then, $p-1$ splitters $f_i = \StringSet{V}[vj - 1]$ are selected for $i = 1,\ldots,p-1$.
This sorting and selection can be implemented trivially by sending all samples to one PE, or using distributed sorting and selection algorithms. In both cases the complete set of splitters is communicated to all PEs.
\end{enumerate}

To prove that buckets are well-balanced, we first show a lemma reformulating the density of samples in a subsequence.

\begin{lemma}\label{lem:string-subarray}
  For $i = 1,\ldots,p$ let $\StringSet{S}'_i = \{s \in \StringSet{S}_i \mid a \leq s \leq b\}$ be an arbitrary contiguous subarray of $\StringSet{S}_i$. If $|\StringSet{S}'_i \cap \StringSet{V}_i| = k$, then $|\StringSet{S}'_i| \leq (k+1) \omega$ with $\omega = |\StringSet{S}_i|/(v+1)$.
\end{lemma}
\begin{proof}
  If $k=0$, then all elements of $\StringSet{S}'_i$ are fully contained between two consecutive sample elements of $\StringSet{V}_i$, thus $|\StringSet{S}'_i| < \omega$.
  If $k=1$, then let $x$ be the element in $\StringSet{S}'_i \cap \StringSet{V}_i$ and we have $\StringSet{S}'_i = \StringSet{S}'_{i,<} \cup \{ x \} \cup \StringSet{S}'_{i,>}$. For $\StringSet{S}'_{i,<}$ and $\StringSet{S}'_{i,>}$ the case $k=0$ applies and thus we have $|\StringSet{S}'_i| \leq (\omega - 1) + 1 + (\omega - 1) \leq 2 \omega$.
  If $k \geq 2$, then we can split $\StringSet{S}'_i$ into $(k+1)$ parts such that the first and last contain at most $\omega-1$ elements and the others exactly $\omega - 1$ between the splitters. Thus $|\StringSet{S}'_i| \leq k ((\omega - 1) + 1) + (\omega - 1) \leq (k+1) \omega$.
\end{proof}

\begin{theorem}\label{thm:string-splitting}
  All buckets $b_j$ contain at most $\frac{n}{p} + \frac{n}{v}$ elements.
\end{theorem}
\begin{proof}
  Let $\mathcal{B}_i^j := \{ s \in \StringSet{S}_i \mid f_{j-1} < s \leq f_j \}$ be the elements in bucket $b_j$ on PE $i$, $\mathcal{V}_i^j := \{ s \in \mathcal{V}_i \mid f_{j-1} < s \leq f_j \}$ the samples therein, and $v_i^j := |\mathcal{V}_i^j|$ their number.
  By definition $|\mathcal{B}_i^j \cap \mathcal{V}_i| = v_i^j$ and by applying Lemma~\ref{lem:string-subarray} we get $|\mathcal{B}_i^j| \leq (v_i^j + 1) \omega$.
  Since $f_{j-1}$ and $f_j$ are globally separated by $v-1$ samples, $\sum_{i=1}^p v_i^j = v$.
  We can now bound $b_j$ by summing over all PEs:
  $|b_j| = \sum_{i=1}^p |\mathcal{B}_i^j| \leq \sum_{i=1}^p (v_i^j + 1) \omega
  = \omega (v + p) = \frac{|\StringSet{S}_i|}{(v+1)} (v + p)
  = \frac{|\StringSet{S}|}{(v+1)p} (v + p)
  < \frac{|\StringSet{S}|}{vp}(v+p)
  = \frac{n}{v} + \frac{n}{p}$\,.
\end{proof}

For character-based partitioning, we have to switch our focus from the string arrays $\StringSet{S}_i$ to the underlying character arrays $\mathcal{C}(\StringSet{S}_i)$.
For simplicity assume that $|\mathcal{C}(\StringSet{S}_i)|$ is divisible by $v + 1$ and let $\omega' = |\mathcal{C}(\StringSet{S}_i)|/(v+1)$.
We furthermore assume $\lmax \leq \omega'$, otherwise strings are very long and too few to draw $v$ samples.

For character-based partitioning, each PE $i$ again chooses $v$ sample strings $\StringSet{V}_i$ from its string set, but this time the strings are regularly sampled such that $\mathcal{C}(\StringSet{S}_i)$ is evenly spaced between them.
For this chose the first strings \emph{starting at or following} the characters at ranks $j \omega' - 1$ in $\mathcal{C}(\StringSet{S}_i)$ for $j = 1,\ldots,v$.
This can be efficiently calculated by keeping an array containing the length of each string while sorting $\StringSet{S}_i$ in Step~1.
As before, the $pv$ sample strings are globally sorted into $\mathcal{V}$ and $p-1$ splitters $f_i$ are selected.

The following lemma states that at most an imbalance of $\lmax$ is introduced due to the shift to the next string.
Using it we can then show a character-based lemma equivalent to Lemma~\ref{lem:string-subarray}.

\begin{lemma}\label{lem:char-partsize}
  If $\lmax \leq \omega'$, then the splitters $\mathcal{V}_i$ selected by character-based partitioning split $\StringSet{S}_i$ into $v+1$ non-empty local buckets $\StringSet{S}_i^j$.
  The number of characters in each bucket $\StringSet{S}_i^j$ is at most $\omega' + \lmax$.
\end{lemma}
\begin{proof}
  We have $\mathcal{S}_i^j := \{ s \in \StringSet{S}_i \mid \mathcal{V}_i[j-1] < s \leq \mathcal{V}_i[j] \}$ assuming $\mathcal{V}_i$ is sorted.
  Since the initially chosen equally spaced characters have a distance of $\omega' \geq \lmax$, the splitters in $\mathcal{V}_i$ are distinct and thus each $\mathcal{S}_i^j$ contains at least the splitter.
  On the other hand, the furthest possible shift from the character-based split point to the next string is $\lmax$, hence each bucket contains at most $\omega' + \lmax$ characters.
\end{proof}

\begin{lemma}\label{lem:char-subarray}
  For $i = 1,\ldots,p$ let $\StringSet{S}'_i = \{s \in \StringSet{S}_i \mid a \leq s \leq b\}$ be an arbitrary contiguous subarray of $\StringSet{S}_i$. If $|\StringSet{S}'_i \cap \StringSet{V}_i| = k$, then $|\mathcal{C}(\StringSet{S}'_i)| \leq (k+1) (\omega' + \lmax)$ with $\omega' = |\mathcal{C}(\StringSet{S}_i)|/(v+1)$.
\end{lemma}
\begin{proof}
  Let $\hat{\omega} := \omega' + \lmax$.
  If $k=0$, then all elements of $\StringSet{S}'_i$ are fully contained in one of the sets $\StringSet{S}'_i$, hence $|\mathcal{C}(\StringSet{S}'_i)| \leq \hat{\omega}$ by Lemma~\ref{lem:char-partsize}.
  The remaining proof for $k=1$ and $k \geq 2$ is analogous to the proof of Lemma~\ref{lem:string-subarray} with $\hat{\omega}$ taking the role of $\omega$ due to Lemma~\ref{lem:char-partsize}.
\end{proof}

With the two lemmas we can reiterate Theorem~\ref{thm:string-splitting} to bound the size of buckets for characters-based partitioning.

\begin{theorem}\label{thm:char-splitting}
  All buckets $b_j$ contain at most $\frac{N}{p} + \frac{N}{v} + (p+v)\lmax$ characters.
\end{theorem}
\begin{proof}
  Applying Lemma~\ref{lem:char-subarray} with the same arguments as in the proof of Theorem~\ref{thm:string-splitting} yields
  $|b_j| = \sum_{i=1}^p |\mathcal{B}_i^j| \leq \sum_{i=1}^p (v_i^j + 1) (\omega' + \lmax) = (\omega' + \lmax) (v + p)
  = \frac{|\mathcal{C}(\StringSet{S}_i)|}{(v+1)} (v + p) + \lmax (v + p)
  = \frac{|\mathcal{C}(\StringSet{S})|}{(v+1)p} (v + p) + \lmax (v + p)
  < \frac{|\mathcal{C}(\StringSet{S})|}{vp}(v+p) + \lmax (v + p)
  = \frac{N}{v} + \frac{N}{p} + \lmax (v + p)$\,.
\end{proof}

\subsection{Data Exchange}

\begin{lemma}
  The data exchange phase of Algorithm~\dMSS (Step~3) with LCP compression has bottleneck communication volume
  $\Oh{(\Nmax + p \lmax) \log \sigma + \nmax \log \lmax}$ bits when character-based sampling is used.
\end{lemma}

\begin{proof}
  The term $\nmax \log \lmax$ stems from the LCP values.  $\Nmax$
  is an obvious upper bound for the string data on each PE.  By Theorem~\ref{thm:char-splitting},
  character-based sampling with $v = \Th{p}$ samples per PE guarantees
  $\Oh{N/p+p\lmax}\leq \Oh{\Nmax+p\lmax}$ characters on the receiving side.
\end{proof}

Note that LCP compression is of no help in establishing non-trivial
worst case bounds on the communication volume. The reason is that
local LCP values may be very short even if every string has long LCPs
with strings located on other PEs. The situation is even worse with
string based sampling since it may happen that some PE gets $n/p$ very
long strings.

\subsection{Overall Analysis of Algorithm~\dMSS}

We now analyze Algorithm~\dMSS with character-based sampling and
using algorithm \hQuick for sorting the sample.
\begin{theorem}\label{thm:dMSS}
  With the notation from Table~\ref{tab:notation},
  Algorithm \dMSS can be implemented to run using
  local work $\Oh{\Nmax+\nmax\log n +p\lmax\log n}$,
  latency $\Oh{\alpha p}$, and bottleneck communication volume
  $\Oh{(\Nmax+p \lmax \log p)\log\sigma}$ bits.
\end{theorem}
Once more, we first interpret this result.  When the input is balanced
with respect to the number of strings and number of characters (i.e.,
$\nmax=\Oh{n/p}$ and $\Nmax=\Oh{N/p}$), and if it is sufficiently
large (i.e., $N=\Om{p^2\lmax}$), we get an algorithm that is as
efficient as we can expect from a method that communicates all the
data. Hence, for large inputs this is a big improvement over \hQuick.
In the worst case, we have no advantage from LCP compression even if
$D\ll N$.  However, by using character-based sampling, we achieve load
balancing guarantees.  Using parallel sorting of the sample saves
us a
factor $p$ in the minimal efficient input size
compared to \cite{fischer2019lightweight} since a
deterministic sampling approach needs samples of
\emph{quadratic} size.
\begin{proof}
  After local sorting (in time $\Oh{\Dmax+\nmax\log\nmax}$),
  each PE samples $v = \Th{p}$ strings locally which have maximal length
  $\lmax$.  These $\Oh{p^2}$ strings are then sorted using algorithm
  \hQuick. By Theorem~\ref{thm:hQuick},
  this incurs local work $\Oh{p\lmax\log n}$, latency $\Oh{\alpha\log^2p}$, and
  communication volume $\Oh{p\lmax\log\sigma\log p}$ bits.
  The
  splitter strings are then gossiped to all PEs. This contributes
  latency $\alpha\log p$ and communication volume $p\lmax\log\sigma$
  bits.

  The local data is then partitioned in time $\Oh{p\log(\nmax) \lmax}$ using binary search.
  By Theorem~\ref{thm:char-splitting}, each of the resulting $p\times p$ messages
  has size $\Oh{\Nmax/p + p\lmax}$. Moreover, no PE receives more than $\Oh{N/p + p\lmax}$ characters.
  Hence, the ensuing all-to-all data exchange
  contributes latency $\Oh{\alpha p}$ and communication volume $\Oh{(\Nmax + p\lmax)\log\sigma}$ bits.
  Finally, the received data is merged in time $\Oh{N/p\log p}$.
  Adding all these terms and making some simplifications yields the claimed result.
\end{proof}

%%%%%%%%%%%%%%%%%%%%%%%%%%%%%%%%%%%%%%%%%%%%%%%%%%%%%%%%%%%%%%%%%%%%%%
\section{Distributed Prefix-Doubling\\ String Merge Sort}\label{s:DPMSS}

We now refine algorithm \dMSS so that it can take advantage of the
case $D\ll N$.  The idea is to find an upper bound for the
distinguishing prefix length of each input string.  We do this as a
Step~$(1+\varepsilon)$ after local sorting (Step~1) but before determining
splitters (Step~2).  The required global communication is
expensive but it pays off in theory and in Section~\ref{s:experiments}
we will see that this algorithm also works well in practice.  We not
only save in communication volume in Step~3 but knowing the
distinguishing prefix lengths also aids (character-based) splitter
determination in finding splitters that balance the actual amount of
work that needs to be done.

Algorithm \dPDSS does not solve exactly the same problem as
Algorithm~\dMSS.  Whereas \dMSS permutes the strings into sorted
order, \dPDSS only computes the permutation without completely
executing it -- it only permutes the distinguishing prefixes (and can
indicate the origin of these prefixes). Note that some applications do
not need the complete information; for example, when string sorting is
used as a subroutine in suffix sorting
\cite{futamura2001parallel,KarSan05,fischer2019lightweight}.  The locally available
information also suffices to build a sorted array of the strings for
pattern search or to build a search tree
\cite{bayer1977prefix,GraLar01}.  The resulting search data structures
support many operations (e.g., counting matches) based on local
information.

\begin{theorem}\label{thm:dPDSS}
  With the notation from Table~\ref{tab:notation},
  Algorithm \dPDSS can be implemented to run using local work
  $\Oh{\Dmax + \nmax\log n}$, latency $\Oh{\alpha p\log\dmax}$, and
  bottleneck communication volume
  $$(1+\varepsilon)\Dmax\log\sigma+\Oh{\nmax\log p+ p\dmax\log\sigma\log p}$$
  bits, in expectation and assuming that the
  all-to-all communication in Step~3 incurs bottleneck communication volume
  $h$ when the maximum sum of local message sizes is $h$.%
  \footnote{It seems to be an open problem whether there is an algorithm
    achieving this. We make this assumption in order to be able to
    concisely work out the impact of the tuning parameter $\epsilon$.}
  The latency can be reduced to $\Oh{\alpha(p + \log p\log\dmax)}$
  when increasing the
  term $\nmax\log p$ in the communication volume to $\nmax\log^2 p$.
\end{theorem}
Again, we interpret the result before
proving it.  Compared to Algorithm~\dMSS, we now achieve local work
and bottleneck communication volume that is close to a worst case lower
bounds if the input is sufficiently large and
sufficiently evenly distributed over the PEs. The price we pay is a
logarithmic factor in the latency which correspondingly increases the
input size that is required to achieve overall efficiency.
\begin{proof}
We only discuss the differences to Algorithm~\dMSS and refer to
Theorem~\ref{thm:dMSS} for remaining details.  The work for Step~1 is
$\Oh{\Dmax+\nmax\log\nmax}$ using any efficient sequential comparison-based
string sorting algorithm.

The analysis of Step~2 is similar to that in Theorem~\ref{thm:dMSS}
except that we are now using samples and splitter strings of length
at most $\dmax$. Also, we do not use the total string lengths as the
basis for sampling but the length of the approximated distinguishing
prefix lengths.
Using Algorithm~\hQuick on the sample now incurs
local work $\Oh{p\dmax\log p}$, latency $\Oh{\alpha\log^2p}$, and
communication volume $\Oh{p\dmax\log\sigma\log p}$ bits.

Refer to Theorem~\ref{thm:distinguishing} for the analysis of Step~$1+\varepsilon$.

The all-to-all exchange in Step~3 incurs latency $\alpha\log p$ and
communication volume $(1+\varepsilon/2)\Dmax\log\sigma$ for those strings
whose prefix length has been successfully approximated within a factor
$1+\varepsilon/2$. We add an additional volume
$\varepsilon\Dmax/2\cdot\log\sigma$ to account for the $o(1)$ term in the
analysis of Step~$1+\varepsilon$ and for the prefix lengths that are
overestimated due to false positives in the duplicate detection. This
works out by setting an appropriate false positive rate $\approx 1/2$.
Overall, we calculate bottleneck communication volume
$(1+\varepsilon)\Dmax$ for Step~3.

In Step~4, the received data is merged in time $\Oh{D/p\log p}$.  Adding
all these terms and making some simplifications yields the claimed
result.
\end{proof}

\begin{figure}[t!]
  \centering%
  \begin{tikzpicture}[
    scale=0.36, yscale=1.153,
    letter/.style={
      font=\ttfamily, anchor=base,
      inner sep=0, minimum width=3.6mm, text height=2.9mm, text depth=0.7mm},
    lcpv/.style={
      anchor=base, inner sep=0,
      font={\small},
    },
    ]

    % \drawstring{y}{prefix}{color}{s,t,r,i,n,g}
    \def\drawstring#1#2#3#4{
      \foreach \c [count=\x from 0] in {#4} {
        \ifthenelse{\x<#2}{
          \node[letter,#3,font=\ttfamily\bfseries] (L#1-\x) at (\x,-#1) {\c};
        }{
          \node[letter,c0] (L#1-\x) at (\x,-#1) {\c};
        }
        \xdef\xmax{\x}
      }
      \draw (-0.5,-#1-0.16) rectangle (\xmax+0.5,-#1+0.66);
      %\node[lcpv] at (-0.9,-#1) {#2};
    }

    \colorlet{c0}{black}
    \colorlet{c1}{blue!80!black}
    \colorlet{c2}{red!80!black}

    \begin{scope}
      \node at (10,+1.55) [anchor=base] {Step 1: sort locally (with LCP array output)};

      \begin{scope}
        \drawstring{0}{0}{c1}{a,l,g,a,e};
        \drawstring{1}{0}{c1}{a,l,p,h,a};
        \drawstring{2}{0}{c1}{a,l,p,s};
        \drawstring{3}{0}{c1}{o,r,d,e,r};
        \draw (-1.5,1) rectangle (6,-3.6);
      \end{scope}
      \begin{scope}[xshift=8cm]
        \drawstring{0}{0}{c1}{a,l,g,o};
        \drawstring{1}{0}{c1}{s,n,o,w};
        \drawstring{2}{0}{c1}{s,o,r,b,e,t};
        \drawstring{3}{0}{c1}{s,o,r,t,e,r};
        \draw (-1.5,1) rectangle (6,-3.6);
      \end{scope}
      \begin{scope}[xshift=16cm]
        \drawstring{0}{0}{c1}{o,r,a,n,g,e};
        \drawstring{1}{0}{c1}{o,r,g,a,n};
        \drawstring{2}{0}{c1}{s,o,r,t,e,d};
        \drawstring{3}{0}{c1}{s,o,u,l};
        \draw (-1.5,1) rectangle (6,-3.6);
      \end{scope}

      \node at (10,-4.55) [anchor=base] {Step $1+\varepsilon$ (depth 1): approximate distinguishing prefix};
      \coordinate (A) at (0, -4.6);
    \end{scope}

    \begin{scope}[shift=(A),yshift=-14mm]
      \begin{scope}
        \drawstring{0}{1}{c1}{a,l,g,a,e};
        \drawstring{1}{1}{c1}{a,l,p,h,a};
        \drawstring{2}{1}{c1}{a,l,p,s};
        \drawstring{3}{1}{c1}{o,r,d,e,r};
        \draw (-1.5,1) rectangle (6,-3.6);
      \end{scope}
      \begin{scope}[xshift=8cm]
        \drawstring{0}{1}{c1}{a,l,g,o};
        \drawstring{1}{1}{c1}{s,n,o,w};
        \drawstring{2}{1}{c1}{s,o,r,b,e,t};
        \drawstring{3}{1}{c1}{s,o,r,t,e,r};
        \draw (-1.5,1) rectangle (6,-3.6);
      \end{scope}
      \begin{scope}[xshift=16cm]
        \drawstring{0}{1}{c1}{o,r,a,n,g,e};
        \drawstring{1}{1}{c1}{o,r,g,a,n};
        \drawstring{2}{1}{c1}{s,o,r,t,e,d};
        \drawstring{3}{1}{c1}{s,o,u,l};
        \draw (-1.5,1) rectangle (6,-3.6);
      \end{scope}

      \node at (10,-4.55) [anchor=base] {Step $1+\varepsilon$ (depth 2): using distributed duplicate detection};
      \coordinate (A) at (0, -4.6);
    \end{scope}

    \begin{scope}[shift=(A),yshift=-14mm]
      \begin{scope}
        \drawstring{0}{2}{c1}{a,l,g,a,e};
        \drawstring{1}{2}{c1}{a,l,p,h,a};
        \drawstring{2}{2}{c1}{a,l,p,s};
        \drawstring{3}{2}{c1}{o,r,d,e,r};
        \draw (-1.5,1) rectangle (6,-3.6);
      \end{scope}
      \begin{scope}[xshift=8cm]
        \drawstring{0}{2}{c1}{a,l,g,o};
        \drawstring{1}{2}{c2}{s,n,o,w};
        \drawstring{2}{2}{c1}{s,o,r,b,e,t};
        \drawstring{3}{2}{c1}{s,o,r,t,e,r};
        \draw (-1.5,1) rectangle (6,-3.6);
      \end{scope}
      \begin{scope}[xshift=16cm]
        \drawstring{0}{2}{c1}{o,r,a,n,g,e};
        \drawstring{1}{2}{c1}{o,r,g,a,n};
        \drawstring{2}{2}{c1}{s,o,r,t,e,d};
        \drawstring{3}{2}{c1}{s,o,u,l};
        \draw (-1.5,1) rectangle (6,-3.6);
      \end{scope}

      \node at (10,-4.55) [anchor=base] {Step $1+\varepsilon$ (depth 4): repeat until all prefixes};
      \coordinate (A) at (0, -4.6);
    \end{scope}

    \begin{scope}[shift=(A),yshift=-14mm]
      \begin{scope}
        \drawstring{0}{4}{c2}{a,l,g,a,e};
        \drawstring{1}{4}{c2}{a,l,p,h,a};
        \drawstring{2}{4}{c2}{a,l,p,s};
        \drawstring{3}{4}{c2}{o,r,d,e,r};
        \draw (-1.5,1) rectangle (6,-3.6);
      \end{scope}
      \begin{scope}[xshift=8cm]
        \drawstring{0}{4}{c2}{a,l,g,o};
        \drawstring{1}{2}{c2}{s,n,o,w};
        \drawstring{2}{4}{c2}{s,o,r,b,e,t};
        \drawstring{3}{4}{c1}{s,o,r,t,e,r};
        \draw (-1.5,1) rectangle (6,-3.6);
      \end{scope}
      \begin{scope}[xshift=16cm]
        \drawstring{0}{4}{c2}{o,r,a,n,g,e};
        \drawstring{1}{4}{c2}{o,r,g,a,n};
        \drawstring{2}{4}{c1}{s,o,r,t,e,d};
        \drawstring{3}{4}{c2}{s,o,u,l};
        \draw (-1.5,1) rectangle (6,-3.6);
      \end{scope}

      \node at (10,-4.55) [anchor=base] {Step $1+\varepsilon$ (depth 8): are unique.};
      \coordinate (A) at (0, -4.6);
    \end{scope}

    \begin{scope}[shift=(A),yshift=-14mm]
      \begin{scope}
        \drawstring{0}{4}{c2}{a,l,g,a,e};
        \drawstring{1}{4}{c2}{a,l,p,h,a};
        \drawstring{2}{4}{c2}{a,l,p,s};
        \drawstring{3}{4}{c2}{o,r,d,e,r};
        \draw (-1.5,1) rectangle (6,-3.6);
      \end{scope}
      \begin{scope}[xshift=8cm]
        \drawstring{0}{4}{c2}{a,l,g,o};
        \drawstring{1}{2}{c2}{s,n,o,w};
        \drawstring{2}{4}{c2}{s,o,r,b,e,t};
        \drawstring{3}{8}{c2}{s,o,r,t,e,r};
        \draw (-1.5,1) rectangle (6,-3.6);
      \end{scope}
      \begin{scope}[xshift=16cm]
        \drawstring{0}{4}{c2}{o,r,a,n,g,e};
        \drawstring{1}{4}{c2}{o,r,g,a,n};
        \drawstring{2}{8}{c2}{s,o,r,t,e,d};
        \drawstring{3}{4}{c2}{s,o,u,l};
        \draw (-1.5,1) rectangle (6,-3.6);
      \end{scope}

      \node at (10,-4.55) [anchor=base] {Step 2: sample regularly: \{ \texttt{alph}, \texttt{sn}, \texttt{orga} \}, };
      \node at (10,-5.55) [anchor=base] {select splitters: \{ \texttt{alph}, \texttt{orga} \}, and find splits.};
      \coordinate (A) at (0, -4.6);
    \end{scope}

    \begin{scope}[shift=(A),yshift=-24mm]
      \tikzset{lcpv/.append style={white}}
      \begin{scope}
        \drawstring{0}{4}{c2}{a,l,g,a,e};
        \drawstring{1}{4}{c2}{a,l,p,h,a};
        \drawstring{2}{4}{c2}{a,l,p,s};
        \drawstring{3}{4}{c2}{o,r,d,e,r};
        \draw (-1.5,1) rectangle (6,-3.6);
        \draw [semithick,green!50!black] (-0.7,-1.25) -- +(6.4,0);
        \draw [semithick,green!50!black] (-0.7,-3.25) -- +(6.4,0);
      \end{scope}
      \begin{scope}[xshift=8cm]
        \drawstring{0}{4}{c2}{a,l,g,o};
        \drawstring{1}{2}{c2}{s,n,o,w};
        \drawstring{2}{4}{c2}{s,o,r,b,e,t};
        \drawstring{3}{8}{c2}{s,o,r,t,e,r};
        \draw (-1.5,1) rectangle (6,-3.6);
        \draw [semithick,green!50!black] (-0.7,-0.2) -- +(6.4,0);
        \draw [semithick,green!50!black] (-0.7,-0.3) -- +(6.4,0);
      \end{scope}
      \begin{scope}[xshift=16cm]
        \drawstring{0}{4}{c2}{o,r,a,n,g,e};
        \drawstring{1}{4}{c2}{o,r,g,a,n};
        \drawstring{2}{8}{c2}{s,o,r,t,e,d};
        \drawstring{3}{4}{c2}{s,o,u,l};
        \draw (-1.5,1) rectangle (6,-3.6);
        \draw [semithick,green!50!black] (-0.7,0.75) -- +(6.4,0);
        \draw [semithick,green!50!black] (-0.7,-1.25) -- +(6.4,0);
      \end{scope}

      \node at (10,-4.55) [anchor=base] {Step 3: all-to-all exchange with LCP compression};
      \coordinate (A) at (0, -4.6);
    \end{scope}

    % \drawstring{y}{prefix}{color}{s,t,r,i,n,g}
    \def\drawstring#1#2#3{
      \foreach \c [count=\x from 0] in {#3} {
        \ifthenelse{\x<#2}{
          \node[letter] (L#1-\x) at (\x,-#1) {\c};
        }{
          \node[letter,black!25] (L#1-\x) at (\x,-#1) {\c};
        }
        \xdef\xmax{\x}
      }
      \draw (-0.5,-#1-0.16) rectangle (\xmax+0.5,-#1+0.66);
      %\node[lcpv] at (-0.9,-#1) {#2};
    }

    \colorlet{c0}{green}
    \colorlet{c2}{red!80!black}
    \colorlet{c3}{black!20}

    \begin{scope}[shift=(A),yshift=-14mm]
      \begin{scope}
        \drawstring{0}{4}{a,l,g,a,e};
        \drawstring{1}{4}{-,-,p,h,a};
        \drawstring{2}{4}{a,l,g,o};
        \draw (-1.5,1) rectangle (6,-4.6);
        \draw [semithick,green!50!black] (-0.7,-1.25) -- +(6.4,0);
        \draw [semithick,green!50!black] (-0.7,-2.25) -- +(6.4,0);
      \end{scope}
      \begin{scope}[xshift=8cm]
        \drawstring{0}{4}{a,l,p,s};
        \drawstring{1}{4}{o,r,d,e,r};
        \drawstring{2}{4}{o,r,a,n,g,e};
        \drawstring{3}{4}{-,-,g,a,n};
        \draw (-1.5,1) rectangle (6,-4.6);
        \draw [semithick,green!50!black] (-0.7,-1.2) -- +(6.4,0);
        \draw [semithick,green!50!black] (-0.7,-1.3) -- +(6.4,0);
      \end{scope}
      \begin{scope}[xshift=16cm]
        \drawstring{0}{2}{s,n,o,w};
        \drawstring{1}{4}{-,o,r,b,e,t};
        \drawstring{2}{6}{-,-,-,t,e,r};
        \drawstring{3}{6}{s,o,r,t,e,d};
        \drawstring{4}{4}{-,-,u,l};
        \draw (-1.5,1) rectangle (6,-4.6);
        \draw [semithick,green!50!black] (-0.7,0.75) -- +(6.4,0);
        \draw [semithick,green!50!black] (-0.7,-2.25) -- +(6.4,0);
      \end{scope}

      \node at (10,-5.55) [anchor=base] {Step 4: multiway merge locally with LCP loser tree};
      \coordinate (A) at (0, -5.6);
    \end{scope}

    \begin{scope}[shift=(A),yshift=-14mm]
      \tikzset{lcpv/.append style={white}}
      \begin{scope}
        \drawstring{0}{4}{a,l,g,a,e};
        \drawstring{1}{4}{a,l,g,o};
        \drawstring{2}{4}{a,l,p,h,a};
        \draw (-1.5,1) rectangle (6,-4.6);
      \end{scope}
      \begin{scope}[xshift=8cm]
        \drawstring{0}{4}{a,l,p,s};
        \drawstring{1}{4}{o,r,a,n,g,e};
        \drawstring{2}{4}{o,r,d,e,r};
        \drawstring{3}{4}{o,r,g,a,n};
        \draw (-1.5,1) rectangle (6,-4.6);
      \end{scope}
      \begin{scope}[xshift=16cm]
        \drawstring{0}{2}{s,n,o,w};
        \drawstring{1}{4}{s,o,r,b,e,t};
        \drawstring{2}{6}{s,o,r,t,e,d};
        \drawstring{3}{6}{s,o,r,t,e,r};
        \drawstring{4}{4}{s,o,u,l};
        \draw (-1.5,1) rectangle (6,-4.6);
      \end{scope}
    \end{scope}

  \end{tikzpicture}%
  \caption{\label{fig:PDMS}Steps of Algorithm \dPDSS shown on example strings.
    String prefixes marked blue are duplicates, while red prefixes are unique.
    In Step 3, only the approximate distinguishing prefix is transmitted, the
    omitted characters are marked gray.%
  }%
\end{figure}

%----------------------------------------------------------------------
\subsection{Approximating Distinguishing Prefix Lengths}\label{ss:distinguishing}

Determining whether a prefix of an input string is a distinguishing
prefix is equivalent to finding out whether there are any duplicates
of it. Duplicate detection is a well studied problem. There is no
known deterministic solution to the problem apart from
communicating the entire prefix. However, we can use randomization.
We calculate hash values (a.k.a. fingerprint) of the prefixes to be
considered and determine which fingerprints are unique. The
corresponding prefixes are now certain to be distinguishing prefixes.
Errors are on the safe side -- two fingerprints may be accidentally
identical which would lead to falsely declaring their corresponding
prefixes to be non-distinguishing. By judiciously choosing the
fingerprint size, by compressing fingerprints, and by iterating the
process with a short fingerprint in the first iteration and a long
fingerprint in the second iteration (where only few candidates
remain), we can do duplicate detection using only $\Oh{\log p}$ bits
of communication volume for each prefix to be checked
\cite{sanders2013communication}.

To approximate the distinguishing prefix length of a string $s$, we
start from some initial guess $\ell_s$ and then let the guessed length
grow geometrically by a factor $(1+\varepsilon)$ in each iteration. With
our default value of $\varepsilon=1$ we get \emph{prefix doubling} which
we use to name sorting algorithm \dPDSS.
Fig.~\ref{fig:PDMS} shows an illustration of \dPDSS and
we fill in the remaining details in the proof of the following theorem.

\begin{theorem}\label{thm:distinguishing}
With the notation from Table~\ref{tab:notation},
distinguishing prefix lengths can be found using local work
$\Oh{\Dmax}$, latency $\Oh{\alpha p\log\dmax}$, and
  bottleneck communication volume
  $\Oh{\nmax\log p}$ bits, in expectation.
  The latency term can be reduced to $\Oh{\alpha\log p\log\dmax}$
  at the price of increasing the
  term $\nmax\log p$ in the communication volume to $\nmax\log^2 p$.
\end{theorem}
\begin{proof}
Determining approximate distinguishing prefix sizes
starts with an initial guess
$\ell=\Th{\ceil{\log p/\log\sigma}}$ bits and iteratively multiplies $\ell$ by a factor
$1+\varepsilon/2$ taking all strings into account whose first $\ell$
characters are not proven to be unique yet.  Hence, the overall number
of iterations is $\log_{1+\varepsilon/2}\dmax=\Oh{\log\dmax}$.
Each
iteration incurs a latency of $\alpha p$ and communication volume
$\Oh{\log p}$ for each string that has not been eliminated yet.
Summing the communication volume for a particular string $s$ with
distinguishing prefix length $\dpre{s}$ yields communication volume
$\Oh{\log p}+o(\dpre{s})$.  Overall, we account communication volume
$\Oh{\nmax\log p}$.

The above discussion assumes that the all-to-all communication of
fingerprints is done by directly delivering them to their destination.
We can reduce the latency of this all-to-all to $\alpha\log p$ by
delivering the data indirectly, e.g., using a hypercube based
all-to-all \cite{SMDD19}. This increases the communication volume by a
factor $\log p$ however.
\end{proof}

Theorem~\ref{thm:distinguishing} may also be useful outside string
sorting algorithms in order to analyze the input with respect to its
distinguishing prefixes.
A simple application might be to choose an algorithm for suffix sorting
based on approximations of $D$  -- when
$D/n$ is small, we can use string
sorting based algorithms, otherwise, more sophisticated algorithms are
better. We might also use this information to choose the difference cover size in an
implementation of the DC algorithm \cite{KarSan05}.

%----------------------------------------------------------------------
\subsection{Average Case Analysis of Algorithm~\dPDSS and Beyond}\label{ss:average}

Neither Algorithm~\dMSS, nor Algorithm~\dPDSS can profit from LCP
compression in the worst case.  This is because there may be inputs
where all input strings have only very short local LCP values but very
long distinguishing prefix lengths due to similar strings on other
PEs. In order to understand why LCP-compression is nevertheless useful
in practice, we now outline an average case analysis.  To keep things
simple let us first consider random bit strings where 0s or 1s are
chosen independently with probability $1/2$.  Among $n$ strings
uniformly distributed over $p$ PEs, the distinguishing prefix lengths
will be about $\log n$. Locally, the LCP values will be about
$\log(n/p)$. Hence, LCP compression saves us $\log(n/p)$ bits per
string. Thus, only about $\log n-\log(n/p)=\log p$ bits actually need to be
transferred.

Therefore, for random inputs, the communication volume of
Algorithm~\dPDSS is dominated by the $\Oh{\log n}$ bits communicated
for LCP-values, string IDs, etc. We now outline how to
obtain an algorithm beyond \dPDSS that reduces this cost by data
compression.  Local LCP-values ``on-the-average'' only differ by a
constant requiring $\Oh{1}$ bits to communicate them using
a combination of difference encoding and variable-bit-length codes. We
also cannot afford to transfer string IDs ($\log n$ bits) or long
associated information. However, we can still view this as a sorting
algorithm with a similar API as Algorithm~\dPDSS: To reconstruct an
output string $s$ and its associated information, a PE remembers from
which PE $i$ string $s$ was received and at which position $j$ in the array of
strings received from PE $i$ it was located. PE $i$ can then be queried for
the suffix and associated information of $s$. This complication also
explains why the logarithm of the number of permutations of the inputs
($\log n!\approx n\log n$, i.e., about $\log n$ bits per input
string), is \emph{not} a lower bound for our view on the sorting
problem -- we do not compute a full permutation but only a data
structure that allows querying this permutation at a cost of
$\Oh{\log n}$ bits of communication per query.

Let us now turn to more general input models.
Assume now that  the characters come from
a random source with entropy $H$.  Distinguishing prefix sizes are
now about $\log_{1/H}n$ and LCP values are $\log_{1/H}n/p$ so that
only $\log_{1/H}p$ characters need to be transmitted. By additionally
compressing those, we can get down to about $\log p$ bits once more.
This argument not only works for random sources where characters are
chosen independently but also, e.g., for Markov chains.

%%%%%%%%%%%%%%%%%%%%%%%%%%%%%%%%%%%%%%%%%%%%%%%%%%%%%%%%%%%%%%%%%%%%%%
\section{Experiments}\label{s:experiments}
\subsection{Inputs}
We now present experiments based on two large real world data sets (\CommonCrawl and \DnaReads)
and a synthetic data set (\dton) with tunable ratio $r=D/N$; see \cite{Schimek19} for details.
In Section~\ref{ss:moreInputs}, we summarize results for further inputs.
The $i$-th string from the \dton input consists of an appropriate number of repetitions of the first character of $\Sigma$ followed by a base $\sigma$ encoding of $i$ followed by further characters to achieve the desired string length (500 in the numbers reported here). Value $r=0$ means that $i$ begins immediately and $r=1$ means that $i$ stands at the end of the string.

% CommonCrawl:
% line_count = 2118603429, char_count = 84237233595, average LCP 23.873
% distingishing prefix = 55524236979, percentage of char_count = 65.914

Input \CommonCrawl consists of the concatenation of the first 200 files from CommonCrawl (2016-40)%
\footnote{\scriptsize\url{commoncrawl.s3.amazonaws.com/crawl-data/CC-MAIN-2016-40/wet.paths.gz}}
This data consists of 82\,GB of text dumps of websites.
Each line of these files represents one input string.
Here we have $D/N = 0.68$, alphabet size 242, average line length 40 characters, and average LCP 23.9 (60\,\% of each line).

% $ zcat *.gz | ./read-packer > dna-reads.txt
% good: 1267961937
% bad: 17087151
% good_size: 123826706585
% replaced: 662819917

% DnaReads (not packed):
% line_count = 1267961937, char_count = 125094668522, average LCP 29.177
% distingishing prefix = 47276864025, percentage of char_count = 37.793

% DnaReadsPacked:
% line_count = 1267961937, char_count = 31733473781, average LCP 6.940
% distingishing prefix = 12354965583, percentage of char_count = 38.934

As an example for small alphabets and bioinformatics applications, we
consider input \DnaReads which consists of reads of DNA sequences from the 1000 Genomes Project\footnote{\scriptsize\url{www.internationalgenome.org/data-portal/sample}}.
Sorting such inputs is relevant as preprocessing for genome assembly or for building indices on the raw data.
We concatenated the low coverage whole genome sequence (WGS) reads from the lexicographically smallest six samples (HG00099, HG00102, HG00107, HG00114, HG00119, HG00121).
We extracted the reads from the FastQ files discarding quality information and concatenated them in lexicographic order of their accession identifier.
Reads containing any other character than \texttt{A}, \texttt{C}, \texttt{G}, and \texttt{T} were dropped.
The resulting data set contains 125\,GB base pairs in 1.27 million read strings with an alphabet size of four and $D/N = 38\,\%$.
On average a DNA read line is 98.7 base pairs long with an average LCP of 29.2 (30\,\% of each line).
Compared to the \CommonCrawl input, \DnaReads has a considerably lower percentage of characters in the LCPs and distinguishing prefix.
This is due to the DNA base pair sequences being more random than text on web pages.

The \CommonCrawl and \DnaReads data was split such that each PE gets about the same number of characters.
The strings from \dton are randomly distributed over the PEs.

\subsection{Hardware}
All experiments were performed on the distributed-memory cluster ForHLR I.
This cluster consists of 512 compute nodes.
Each of these nodes contains two 10-core Intel Xeon
processors E5-2670 v2 (Sandy Bridge) with a clock speed of 2.5\,GHz and
have 10$\times$256\,KB of level 2 cache and 25\,MB level 3 cache.  Each node
possesses 64\,GB of main memory and an adapter to connect to the
InfiniBand 4X FDR interconnect.%
\footnote{\scriptsize\url{wiki.scc.kit.edu/hpc/index.php/ForHLR_-_Hardware_and_Architecture}}
Intel MPI Library 2018 was used as implementation of the MPI standard.
All programs were compiled with GCC 8.2.0 and optimization flags
\texttt{-O3} and \texttt{-march=native}. We create one MPI process on each available core, i.e., hardware threads are not used.

\begin{figure*}
  \pgfplotscreateplotcyclelist{my_color}{%
    plotcolor9, every mark/.append style={solid,scale=1.1}, mark=Mercedes star flipped \\%
    plotcolor0, every mark/.append style={solid,scale=1.0}, mark=x \\%
    plotcolor1, every mark/.append style={solid,scale=0.9}, mark=o \\%
    plotcolor2, every mark/.append style={solid,scale=1.0}, mark=diamond \\%
    plotcolor3, every mark/.append style={solid,scale=1.0}, mark=triangle \\%
    plotcolor4, every mark/.append style={solid,scale=1.0,rotate=180}, mark=triangle \\%
  }
  \pgfplotsset{
    myPlot1/.style={
      myPlot,
      width=46mm,height=60mm,
      ymin=-1.00, ymax=19,
      cycle list name={my_color},
      log x ticks with fixed point,
      xtick={20,40,80,160,320,640,1280,2560},
      xticklabel style={rotate=90, anchor=east},
      legend columns=6,
      %transpose legend,
    },
    myPlot2/.style={
      myPlot1,
      xticklabels={,,},
    },
  }

  \centering

  \hspace{8ex}
  \begin{tikzpicture}[trim axis left]
    \begin{semilogxaxis}[myPlot2,
      title={$\frac{D}{N} = 0$},
      ylabel={time (s)},
      legend to name={leg:d2n_weak},
      ]

      %% MULTIPLOT(name|ptitle)
      %% SELECT *, name AS ptitle FROM d2n_view WHERE r = 0.0
      \addplot coordinates { (20,1.68186) (40,2.29341) (80,2.66636) (160,2.91365) (320,3.27234) (640,4.25227) (1280,7.51084) };
      \addlegendentry{\kurpicz};
      \addplot coordinates { (20,5.95345) (40,7.54681) (80,10.4359) (160,13.0844) (320,15.137) (640,17.2183) (1280,20.5372) };
      \addlegendentry{\hQuick};
      \addplot coordinates { (20,1.09004) (40,1.72091) (80,2.0655) (160,2.31644) (320,2.56471) (640,2.86422) (1280,3.98063) };
      \addlegendentry{\MSSimpleS};
      \addplot coordinates { (20,1.47596) (40,2.10313) (80,2.44087) (160,2.69991) (320,2.94475) (640,3.27302) (1280,4.47023) };
      \addlegendentry{\MSLCPS};
      \addplot coordinates { (20,0.462157) (40,0.516935) (80,0.563829) (160,0.584205) (320,0.66917) (640,0.750515) (1280,0.895206) };
      \addlegendentry{\PDGolombS};
      \addplot coordinates { (20,0.435868) (40,0.497253) (80,0.543658) (160,0.581725) (320,0.622036) (640,0.732387) (1280,0.876076) };
      \addlegendentry{\PDNoGolombS};

    \end{semilogxaxis}
  \end{tikzpicture}
  \hfill%
  \begin{tikzpicture}
    \begin{semilogxaxis}[myPlot2,
      title={$\frac{D}{N} = 0.25$}, yticklabel=\empty]

      %% MULTIPLOT(name)
      %% SELECT * FROM d2n_view WHERE r = 0.25
      \addplot coordinates { (20,2.50764) (40,3.173) (80,3.63177) (160,3.95237) (320,4.62924) (640,6.32796) (1280,11.9554) };
      \addlegendentry{name=\\kurpicz};
      \addplot coordinates { (20,7.24456) (40,8.98775) (80,12.0366) (160,14.7728) (320,16.8422) (640,18.9669) (1280,22.2726) };
      \addlegendentry{name=\\hQuick};
      \addplot coordinates { (20,1.97247) (40,2.64647) (80,3.04982) (160,3.3626) (320,3.72754) (640,4.17839) (1280,5.31962) };
      \addlegendentry{name=\\MSSimpleS};
      \addplot coordinates { (20,1.93808) (40,2.43646) (80,2.70236) (160,2.93428) (320,3.18128) (640,3.53739) (1280,4.60605) };
      \addlegendentry{name=\\MSLCPS};
      \addplot coordinates { (20,1.25717) (40,1.33083) (80,1.34657) (160,1.40428) (320,1.48392) (640,1.62181) (1280,1.86595) };
      \addlegendentry{name=\\PDGolombS};
      \addplot coordinates { (20,1.25873) (40,1.31279) (80,1.33012) (160,1.39952) (320,1.47916) (640,1.6053) (1280,1.858) };
      \addlegendentry{name=\\PDNoGolombS};

      \legend{}
    \end{semilogxaxis}
  \end{tikzpicture}
  \hfill%
  \begin{tikzpicture}
    \begin{semilogxaxis}[myPlot2,
      title={$\frac{D}{N} = 0.5$}, yticklabel=\empty]

      %% MULTIPLOT(name)
      %% SELECT * FROM d2n_view WHERE r = 0.50
      \addplot coordinates { (20,3.35912) (40,4.07421) (80,4.5833) (160,5.0073) (320,6.0175) (640,8.44688) (1280,16.5044) };
      \addlegendentry{name=\\kurpicz};
      \addplot coordinates { (20,8.6207) (40,10.5517) (80,13.6544) (160,16.5844) (320,18.7044) (640,20.8892) (1280,24.2457) };
      \addlegendentry{name=\\hQuick};
      \addplot coordinates { (20,2.90336) (40,3.61255) (80,4.06268) (160,4.41434) (320,4.86028) (640,5.3965) (1280,6.66228) };
      \addlegendentry{name=\\MSSimpleS};
      \addplot coordinates { (20,2.42654) (40,2.76506) (80,2.97717) (160,3.13746) (320,3.40862) (640,3.61024) (1280,4.34119) };
      \addlegendentry{name=\\MSLCPS};
      \addplot coordinates { (20,2.05441) (40,2.09996) (80,2.12557) (160,2.20416) (320,2.31111) (640,2.46239) (1280,2.75235) };
      \addlegendentry{name=\\PDGolombS};
      \addplot coordinates { (20,2.02638) (40,2.09461) (80,2.13458) (160,2.17763) (320,2.29621) (640,2.44794) (1280,2.74166) };
      \addlegendentry{name=\\PDNoGolombS};

      \legend{}
    \end{semilogxaxis}
  \end{tikzpicture}
  \hfill%
  \begin{tikzpicture}
    \begin{semilogxaxis}[myPlot2,
      title={$\frac{D}{N} = 0.75$}, yticklabel=\empty]

      %% MULTIPLOT(name)
      %% SELECT * FROM d2n_view WHERE r = 0.75
      \addplot coordinates { (20,4.21789) (40,4.97285) (80,5.54216) (160,6.04277) (320,7.27188) (640,10.5287) (1280,20.8205) };
      \addlegendentry{name=\\kurpicz};
      \addplot coordinates { (20,10.0435) (40,12.1236) (80,15.3739) (160,18.4317) (320,20.6093) (640,22.8554) (1280,26.3685) };
      \addlegendentry{name=\\hQuick};
      \addplot coordinates { (20,3.81788) (40,4.57519) (80,5.08588) (160,5.52219) (320,6.00261) (640,6.63849) (1280,7.92202) };
      \addlegendentry{name=\\MSSimpleS};
      \addplot coordinates { (20,2.94884) (40,3.12349) (80,3.26435) (160,3.39165) (320,3.58039) (640,3.57646) (1280,4.01409) };
      \addlegendentry{name=\\MSLCPS};
      \addplot coordinates { (20,2.78807) (40,2.97495) (80,3.09924) (160,3.29609) (320,3.42643) (640,3.41127) (1280,3.91267) };
      \addlegendentry{name=\\PDGolombS};
      \addplot coordinates { (20,2.78188) (40,2.96724) (80,3.10057) (160,3.28027) (320,3.43165) (640,3.41974) (1280,3.88755) };
      \addlegendentry{name=\\PDNoGolombS};

      \legend{}
    \end{semilogxaxis}
  \end{tikzpicture}
  \hfill%
  \begin{tikzpicture}
    \begin{semilogxaxis}[myPlot2,
      title={$\frac{D}{N} = 1.0$}, yticklabel=\empty]

      %% MULTIPLOT(name)
      %% SELECT * FROM d2n_view WHERE r = 1.00
      \addplot coordinates { (20,5.08534) (40,5.8772) (80,6.49925) (160,7.14325) (320,8.55695) (640,12.572) (1280,25.1766) };
      \addlegendentry{name=\\kurpicz};
      \addplot coordinates { (20,11.441) (40,13.6775) (80,17.0041) (160,20.2343) (320,22.5007) (640,24.8315) (1280,28.2393) };
      \addlegendentry{name=\\hQuick};
      \addplot coordinates { (20,4.71975) (40,5.52648) (80,6.0699) (160,6.54026) (320,7.13428) (640,7.65992) (1280,9.12173) };
      \addlegendentry{name=\\MSSimpleS};
      \addplot coordinates { (20,3.37567) (40,3.3903) (80,3.42504) (160,3.48463) (320,3.57187) (640,3.68151) (1280,3.91453) };
      \addlegendentry{name=\\MSLCPS};
      \addplot coordinates { (20,3.22066) (40,3.2238) (80,3.26644) (160,3.33258) (320,3.4296) (640,3.59186) (1280,3.83827) };
      \addlegendentry{name=\\PDGolombS};
      \addplot coordinates { (20,3.21836) (40,3.23808) (80,3.28697) (160,3.32424) (320,3.45889) (640,3.60108) (1280,3.83414) };
      \addlegendentry{name=\\PDNoGolombS};

      \legend{}
    \end{semilogxaxis}
  \end{tikzpicture}

  \pgfplotsset{
    myPlot1/.append style={
      ymin=-40.00, ymax=650,
      height=45mm,
      xlabel={number of PEs},
      xlabel style={yshift=1ex},
    }
  }
  \vspace{-1ex}

  %% SQL CREATE OR REPLACE TEMPORARY VIEW d2n_weak_comm AS
  %% SELECT algo, nprocs, dton, size, SUM(value) AS value, iteration, spec
  %% FROM (SELECT algo, "numberProcessors" AS nprocs, "dToNRatio" AS dton, size, value, iteration, phase,
  %%   algo || '_' || COALESCE("samplePolicy", 'x') || '_' || COALESCE("ByteEncoder", 'x') || '_' || COALESCE("GolombEncoding", 'x') AS spec
  %%   FROM d2n_weak WHERE operation='commVolume' AND iteration != 0 AND phase != 'none') a
  %% GROUP BY algo, nprocs, dton, size, iteration, spec

  %% SQL CREATE OR REPLACE TEMPORARY VIEW d2n_view AS
  %% SELECT exp1.name, nprocs AS x, AVG(value / (size * nprocs)) AS y, d.dton AS r
  %%   FROM d2n_weak_comm d
  %% LEFT JOIN exp1 ON exp1.algo = d.spec
  %% WHERE exp1.order != 0
  %% GROUP BY exp1.order, exp1.name, nprocs, dton ORDER BY exp1.order, exp1.name, x, y

  \hspace{8ex}
  \begin{tikzpicture}[trim axis left]
    \begin{semilogxaxis}[myPlot1,
      ylabel={bytes sent per string},
      legend to name={leg:d2n_weak},
      ]

      %% MULTIPLOT(name|ptitle)
      %% SELECT *, name AS ptitle FROM d2n_view WHERE r = 0.0
      \addplot coordinates { (20,501.763) (40,504.129) (80,513.672) (160,551.995) (320,705.594) (640,1320.61) (1280,3781.89) };
      \addlegendentry{\kurpicz};
      \addplot coordinates { (20,1102.25) (40,1352.72) (80,1603.22) (160,1853.71) (320,2104.23) (640,2354.73) (1280,2605.23) };
      \addlegendentry{\hQuick};
      \addplot coordinates { (20,501.114) (40,501.266) (80,501.611) (160,502.382) (320,504.083) (640,507.811) (1280,515.917) };
      \addlegendentry{\MSSimpleS};
      \addplot coordinates { (20,498.712) (40,498.072) (80,498.597) (160,499.557) (320,501.495) (640,504.479) (1280,512.822) };
      \addlegendentry{\MSLCPS};
      \addplot coordinates { (20,12.4335) (40,11.6489) (80,11.8464) (160,12.0753) (320,12.3928) (640,11.8214) (1280,12.4274) };
      \addlegendentry{\PDGolombS};
      \addplot coordinates { (20,14.6043) (40,13.8188) (80,14.0146) (160,14.2399) (320,14.5505) (640,13.9656) (1280,14.5443) };
      \addlegendentry{\PDNoGolombS};

    \end{semilogxaxis}
  \end{tikzpicture}
  \hfill%
  \begin{tikzpicture}
    \begin{semilogxaxis}[myPlot1, yticklabel=\empty]

      %% MULTIPLOT(name)
      %% SELECT * FROM d2n_view WHERE r = 0.25
      \addplot coordinates { (20,501.763) (40,504.129) (80,513.672) (160,551.995) (320,705.594) (640,1320.61) (1280,3781.89) };
      \addlegendentry{name=\\kurpicz};
      \addplot coordinates { (20,1102.21) (40,1352.76) (80,1603.24) (160,1853.68) (320,2104.22) (640,2354.71) (1280,2605.22) };
      \addlegendentry{name=\\hQuick};
      \addplot coordinates { (20,501.114) (40,501.267) (80,501.611) (160,502.381) (320,504.084) (640,507.812) (1280,515.917) };
      \addlegendentry{name=\\MSSimpleS};
      \addplot coordinates { (20,378.717) (40,379.081) (80,379.616) (160,380.596) (320,382.572) (640,386.629) (1280,395.123) };
      \addlegendentry{name=\\MSLCPS};
      \addplot coordinates { (20,20.4735) (40,20.7381) (80,21.0444) (160,21.5108) (320,22.3439) (640,23.884) (1280,26.8777) };
      \addlegendentry{name=\\PDGolombS};
      \addplot coordinates { (20,22.644) (40,22.9078) (80,23.2122) (160,23.675) (320,24.5014) (640,26.0279) (1280,28.9945) };
      \addlegendentry{name=\\PDNoGolombS};

      \legend{}
    \end{semilogxaxis}
  \end{tikzpicture}
  \hfill%
  \begin{tikzpicture}
    \begin{semilogxaxis}[myPlot1, yticklabel=\empty]

      %% MULTIPLOT(name)
      %% SELECT * FROM d2n_view WHERE r = 0.50
      \addplot coordinates { (20,501.763) (40,504.129) (80,513.672) (160,551.995) (320,705.594) (640,1320.61) (1280,3781.89) };
      \addlegendentry{name=\\kurpicz};
      \addplot coordinates { (20,1102.18) (40,1352.7) (80,1603.23) (160,1853.71) (320,2104.19) (640,2354.72) (1280,2605.25) };
      \addlegendentry{name=\\hQuick};
      \addplot coordinates { (20,501.114) (40,501.266) (80,501.612) (160,502.382) (320,504.083) (640,507.811) (1280,515.917) };
      \addlegendentry{name=\\MSSimpleS};
      \addplot coordinates { (20,254.722) (40,255.091) (80,255.636) (160,256.635) (320,258.651) (640,262.788) (1280,271.441) };
      \addlegendentry{name=\\MSLCPS};
      \addplot coordinates { (20,24.5084) (40,24.8181) (80,25.2237) (160,25.9092) (320,27.221) (640,29.8003) (1280,35.0362) };
      \addlegendentry{name=\\PDGolombS};
      \addplot coordinates { (20,26.679) (40,26.9876) (80,27.3915) (160,28.0733) (320,29.3784) (640,31.9443) (1280,37.1526) };
      \addlegendentry{name=\\PDNoGolombS};

      \legend{}
    \end{semilogxaxis}
  \end{tikzpicture}
  \hfill%
  \begin{tikzpicture}
    \begin{semilogxaxis}[myPlot1, yticklabel=\empty]

      %% MULTIPLOT(name)
      %% SELECT * FROM d2n_view WHERE r = 0.75
      \addplot coordinates { (20,501.763) (40,504.129) (80,513.672) (160,551.995) (320,705.594) (640,1320.61) (1280,3781.89) };
      \addlegendentry{name=\\kurpicz};
      \addplot coordinates { (20,1102.16) (40,1352.74) (80,1603.17) (160,1853.7) (320,2104.22) (640,2354.7) (1280,2605.22) };
      \addlegendentry{name=\\hQuick};
      \addplot coordinates { (20,501.114) (40,501.266) (80,501.611) (160,502.382) (320,504.083) (640,507.812) (1280,515.917) };
      \addlegendentry{name=\\MSSimpleS};
      \addplot coordinates { (20,129.727) (40,130.101) (80,130.655) (160,131.675) (320,133.731) (640,137.948) (1280,146.761) };
      \addlegendentry{name=\\MSLCPS};
      \addplot coordinates { (20,129.738) (40,130.124) (80,130.7) (160,131.764) (320,133.908) (640,138.303) (1280,147.468) };
      \addlegendentry{name=\\PDGolombS};
      \addplot coordinates { (20,129.738) (40,130.123) (80,130.7) (160,131.764) (320,133.908) (640,138.301) (1280,147.468) };
      \addlegendentry{name=\\PDNoGolombS};

      \legend{}
    \end{semilogxaxis}
  \end{tikzpicture}
  \hfill%
  \begin{tikzpicture}
    \begin{semilogxaxis}[myPlot1, yticklabel=\empty]

      %% MULTIPLOT(name)
      %% SELECT * FROM d2n_view WHERE r = 1.00
      \addplot coordinates { (20,501.763) (40,504.129) (80,513.672) (160,551.995) (320,705.594) (640,1320.61) (1280,3781.89) };
      \addlegendentry{name=\\kurpicz};
      \addplot coordinates { (20,1102.21) (40,1352.71) (80,1603.25) (160,1853.73) (320,2104.25) (640,2354.71) (1280,2605.22) };
      \addlegendentry{name=\\hQuick};
      \addplot coordinates { (20,501.114) (40,501.266) (80,501.611) (160,502.382) (320,504.083) (640,507.811) (1280,515.917) };
      \addlegendentry{name=\\MSSimpleS};
      \addplot coordinates { (20,4.73167) (40,5.11136) (80,5.67566) (160,6.71539) (320,8.81071) (640,13.1085) (1280,22.0811) };
      \addlegendentry{name=\\MSLCPS};
      \addplot coordinates { (20,4.74329) (40,5.13387) (80,5.7202) (160,6.80443) (320,8.98839) (640,13.4622) (1280,22.7881) };
      \addlegendentry{name=\\PDGolombS};
      \addplot coordinates { (20,4.74311) (40,5.13341) (80,5.71978) (160,6.80386) (320,8.98788) (640,13.4618) (1280,22.7879) };
      \addlegendentry{name=\\PDNoGolombS};

      \legend{}
    \end{semilogxaxis}
  \end{tikzpicture}

  \medskip
  \ref{leg:d2n_weak}
  \caption{\label{fig:dtonWeak}%
    Running times and bytes sent per string for the weak-scaling experiment with five generated $\frac{D}{N}$ inputs with 500\,000 strings of length 500 per PE.}

\end{figure*}
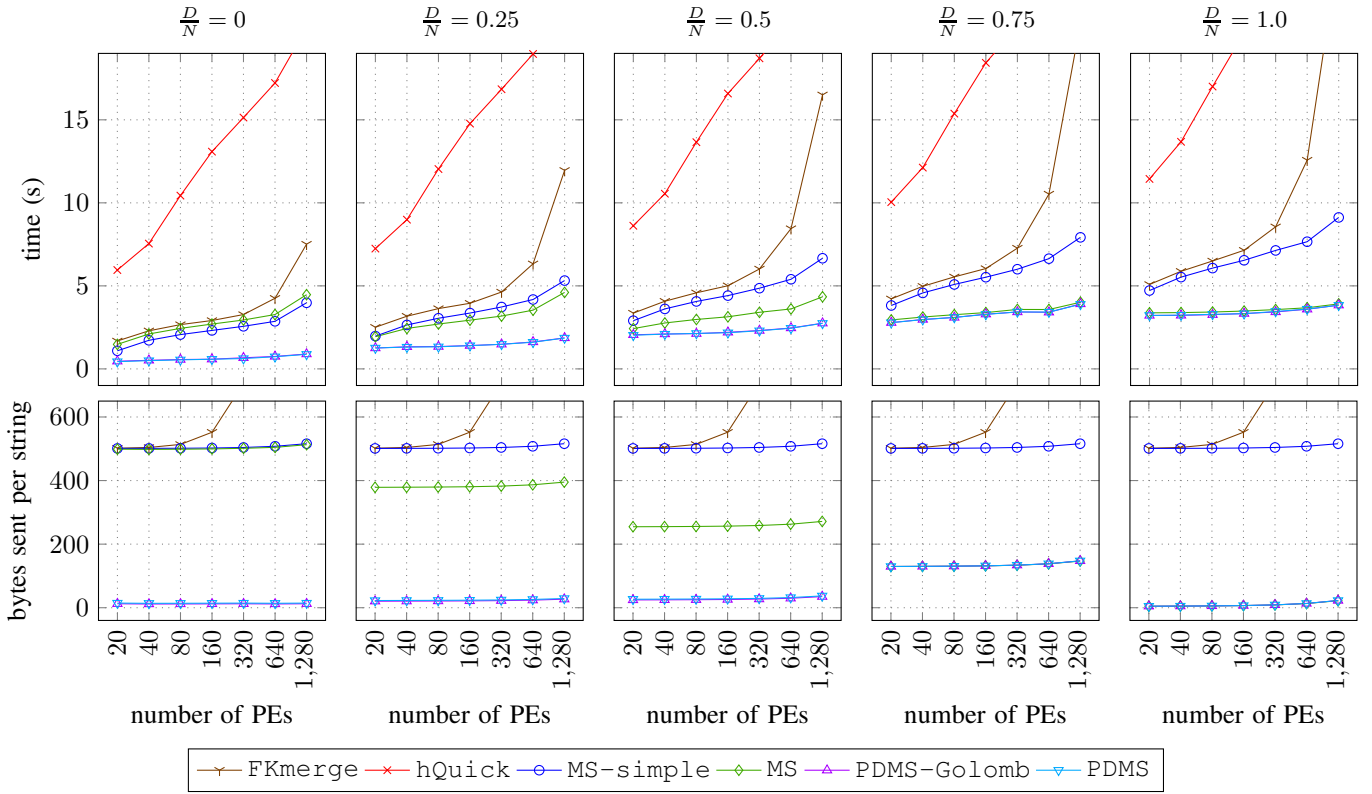

\subsection{Algorithms}
We compare the following algorithms:

\kurpicz is the distributed multiway string mergesort of
Fischer and Kurpicz \cite{fischer2019lightweight}; see also
Section~\ref{ss:dmm}. This is the only distributed-memory string sorter that we could find.

\hQuick: As an example for a fast atomic sorting algorithm we use our
adaptation of distributed hypercube atomic
Quicksort by Axtmann et al. \cite{axtmann2017robust}. We adapted its
original
implementation \cite{KaDiS}
by replacing point-to-point communication of fixed length with
point-to-point communication of variable length. See also Section~\ref{ss:hQuick}.

\MSSimpleS is our Distributed String Merge Sort from
Section~\ref{s:DMSS} with no LCP related optimizations at all.

\MSLCPS is our Algorithm \dMSS with LCP compression.

\PDGolombS is an implementation of our Distributed Prefix-Doubling
String Merge Sort (\dPDSS) from Section~\ref{s:DPMSS} that uses Golomb coding
for communicating hash values.

\PDNoGolombS is the same as \PDGolombS except without using Golomb compression
for communicating hash values.

All these algorithms use string based sampling.
Our C++ implementations of all these algorithms is available as open source from \url{https://github.com/mschimek/distributed-string-sorting}.

\subsection{Results}

Fig.~\ref{fig:dtonWeak} shows a weak scaling experiment with 250\,MB of
data on each core using different ratios $D/N$.
As expected, the atomic sorting
algorithm \hQuick is outclassed by the string sorting algorithms.
The only previous distributed string sorter \kurpicz works well up to
320~cores (16~nodes) but scalability then quickly deteriorates. We
attribute this to high communication costs and a bottleneck due to
centralized sorting of samples.  This is also consistent with the
increasing communication volume observed in the lower part of the
plot that shows communication volume.  Already the most simple variant
of our \dMSS algorithm \MSSimpleS consistently outperforms
\kurpicz and \hQuick, and scales reasonably well -- the execution time with 64
nodes (1280 cores) is only about twice that of the execution time with
2 nodes (40 cores). This ratio gets smaller as $D/N$ grows since there
is more internal work to be done. Enabling the LCP optimization in
\MSLCPS yields further consistent improvements.  Not surprisingly,
the advantages get more pronounced with increasing $D/N$ since this
implies longer common prefixes. The prefix doubling algorithms
(\PDGolombS and \PDNoGolombS) yield a further large improvement when
$D/N$ is not too large because we get a large saving in communication
volume.  For large $D/N$, the prefix doubling yields no useful bounds
on the distinguishing prefix length and hence the moderate overheads
for finding these values makes the algorithms slightly slower than
\MSLCPS.  Using or not using Golomb compression of hash values is of little
consequence on running time. We see on the lower part of the figure
that it also has little influence on communication volume. Apparently,
the communication overhead for finding distinguishing prefixes is
rather small anyway.  On the largest configuration (1280 PEs, 64
nodes, 320\,GB of data) the best shown algorithm is 5.3--8.6 times
faster than \kurpicz.

A look at the communication volumes in the lower part of the plot
underlines the great communication efficiency of combining LCP
compression with prefix doubling (algorithms \PDGolombS and
\PDNoGolombS). Using LCP compression only (Algorithm \MSLCPS) is only
effective when the LCPs are long.

Fig.~\ref{fig:crawlStrong}, left panel,
shows strong scaling results for the \CommonCrawl instance.  Here we cannot
show results for the competing code \kurpicz since it crashes.
Apparently it does not correctly handle inputs with many repeated
input strings.  The ranking of the remaining algorithms is similar as
for the \dton-instances. For $p\geq 480$, the algorithms based on prefix doubling are
5.4--6.1 times faster than \hQuick and \MSLCPS is a factor 4.5--4.6 faster.
The algorithms with LCP compression are 2.6--3.5 times faster than \MSSimpleS.
This indicates that the LCP optimizations are very effective
for the \CommonCrawl-instance while prefix doubling does not help here. This
is consistent with the large $D/N$-ratio of 0.68 for this instance where
prefix doubling cannot be effective.
The running
times keep going down until the largest configuration tried. However,
efficiency is rapidly deteriorating. The reason for the difference to
the above weak scaling result may be that the \CommonCrawl-instance is a
factor four smaller for $p=1280$. Hence, experiments with larger real
world inputs are interesting topics for future work.

Fig.~\ref{fig:crawlStrong}, right panel, shows the corresponding results
for the \DnaReads\ input.  Here, algorithms \MSLCPS and \MSSimpleS
are slightly faster than the prefix doubling algorithms, despite
considerable savings in communication volume. Algorithm \kurpicz
works now but does not scale so well.

\begin{figure*}
  \pgfplotscreateplotcyclelist{my_color}{%
    %plotcolor9, every mark/.append style={solid,scale=1.1}, mark=Mercedes star flipped \\%
    plotcolor0, every mark/.append style={solid,scale=1.0}, mark=x \\%
    plotcolor1, every mark/.append style={solid,scale=0.9}, mark=o \\%
    plotcolor2, every mark/.append style={solid,scale=1.0}, mark=diamond \\%
    plotcolor3, every mark/.append style={solid,scale=1.0}, mark=triangle \\%
    plotcolor4, every mark/.append style={solid,scale=1.0,rotate=180}, mark=triangle \\%
  }
  \pgfplotsset{
    myPlot1/.style={
      myPlot,
      width=90mm,height=52mm,
      cycle list name={my_color},
      log x ticks with fixed point,
      xtick={20,40,80,160,320,480,640,960,1280,2560},
      xticklabel style={rotate=90, anchor=east},
      legend columns=6,
      %transpose legend,
      xlabel={number of PEs},
      xlabel style={yshift=1.5ex},
    },
    myPlot2/.style={
      myPlot1,
      xticklabels={,,},
      xlabel={},
    },
  }

  \centering
  \begin{minipage}{.5\textwidth}
  \hfill%
  \begin{tikzpicture}
    \begin{semilogxaxis}[myPlot2,
      title={\CommonCrawl},
      ylabel={time (s)},
      ymin=0.00, ymax=32,
      legend to name={leg:cc_strong}]

      %% MULTIPLOT(ptitle|ptitle)
      %% SELECT x, y, ptitle FROM cc_view c
      %% ORDER BY c.order, x, y
      \addplot coordinates { (320,35.6809) (480,32.9128) (640,29.7826) (960,24.75) (1280,22.7489) };
      \addlegendentry{\hQuick};
      \addplot coordinates { (160,20.5926) (320,17.4044) (480,15.4007) (640,16.157) (960,14.0163) (1280,14.6748) };
      \addlegendentry{\MSSimpleS};
      \addplot coordinates { (160,14.402) (320,9.13729) (480,7.38794) (640,6.40026) (960,5.41531) (1280,4.96395) };
      \addlegendentry{\MSLCPS};
      \addplot coordinates { (160,14.3727) (320,8.37111) (480,5.95796) (640,4.85112) (960,4.4463) (1280,4.24431) };
      \addlegendentry{\PDGolombS};
      \addplot coordinates { (160,14.2329) (320,8.27867) (480,5.96323) (640,4.87924) (960,4.42829) (1280,4.20009) };
      \addlegendentry{\PDNoGolombS};

    \end{semilogxaxis}
  \end{tikzpicture}

  \vspace{-1ex}%
  \hfill%
  \begin{tikzpicture}
    \begin{semilogxaxis}[myPlot1,
      height=44mm,
      ymin=15, ymax=53,
      ylabel={bytes sent per string}]

      %% MULTIPLOT(name)
      %% SELECT exp2.name, nprocs AS x, AVG(value / '2074964317') AS y, d.dton AS r
      %%   FROM cc_strong_comm d
      %%   LEFT JOIN exp2 ON exp2.algo = d.spec
      %% WHERE exp2.order != 0
      %% GROUP BY exp2.order, exp2.name, nprocs, dton ORDER BY exp2.order, exp2.name, x, y
      \addplot coordinates { (320,167.059) (480,177.112) (640,186.689) (960,196.88) (1280,205.911) };
      \addlegendentry{name=\\hQuick};
      \addplot coordinates { (160,39.5804) (320,39.5957) (480,39.6228) (640,39.6628) (960,39.7815) (1280,39.956) };
      \addlegendentry{name=\\MSSimpleS};
      \addplot coordinates { (160,22.3579) (320,23.0223) (480,23.2155) (640,23.3661) (960,23.6409) (1280,23.9368) };
      \addlegendentry{name=\\MSLCPS};
      \addplot coordinates { (160,19.9725) (320,20.9082) (480,21.2629) (640,21.5737) (960,22.1844) (1280,22.8795) };
      \addlegendentry{name=\\PDGolombS};
      \addplot coordinates { (160,22.6172) (320,23.7908) (480,24.2026) (640,24.5327) (960,25.1297) (1280,25.7602) };
      \addlegendentry{name=\\PDNoGolombS};

      \legend{}
    \end{semilogxaxis}
  \end{tikzpicture}

  \end{minipage}%
  \pgfplotscreateplotcyclelist{my_color}{%
    plotcolor9, every mark/.append style={solid,scale=1.1}, mark=Mercedes star flipped \\%
    plotcolor0, every mark/.append style={solid,scale=1.0}, mark=x \\%
    plotcolor1, every mark/.append style={solid,scale=0.9}, mark=o \\%
    plotcolor2, every mark/.append style={solid,scale=1.0}, mark=diamond \\%
    plotcolor3, every mark/.append style={solid,scale=1.0}, mark=triangle \\%
    plotcolor4, every mark/.append style={solid,scale=1.0,rotate=180}, mark=triangle \\%
  }%
  \pgfplotsset{
    myPlot1/.append style={
      cycle list name={my_color},
    },
  }%
  \begin{minipage}{.5\textwidth}
  \hfill%
  \begin{tikzpicture}
    \begin{semilogxaxis}[myPlot2,
      title={\DnaReads},
      ylabel={time (s)},
      ymin=0.00, ymax=21,
      legend to name={leg:dna_strong}]

      %% MULTIPLOT(ptitle|ptitle)
      %% SELECT x, y, ptitle FROM dna_view c
      %% WHERE x < 2000
      %% ORDER BY c.order, x, y
      \addplot coordinates { (320,11.6735) (480,8.90505) (640,7.49175) (960,6.80904) (1280,6.53329) };
      \addlegendentry{\kurpicz};
      \addplot coordinates { (320,31.643) (480,28.166) (640,17.4096) (960,16.1212) (1280,9.51927) };
      \addlegendentry{\hQuick};
      \addplot coordinates { (160,13.9399) (320,7.54662) (480,5.41524) (640,4.3328) (960,3.47782) (1280,2.59856) };
      \addlegendentry{\MSSimpleS};
      \addplot coordinates { (160,14.4356) (320,7.77576) (480,5.54417) (640,4.30537) (960,3.40826) (1280,2.58642) };
      \addlegendentry{\MSLCPS};
      \addplot coordinates { (160,18.857) (320,9.8398) (480,6.91774) (640,5.13906) (960,3.99573) (1280,3.34914) };
      \addlegendentry{\PDGolombS};
      \addplot coordinates { (160,18.4337) (320,9.69519) (480,6.78517) (640,5.12081) (960,3.96693) (1280,3.33055) };
      \addlegendentry{\PDNoGolombS};

    \end{semilogxaxis}
  \end{tikzpicture}

  \vspace{-1ex}%
  \hfill%
  \begin{tikzpicture}
    \begin{semilogxaxis}[myPlot1,
      height=44mm,
      ymin=31, ymax=115,
      ylabel={bytes sent per string}]

      %% MULTIPLOT(name)
      %% SELECT exp2.name, nprocs AS x, AVG(value / '1267961937') AS y, d.dton AS r
      %%   FROM dna_strong_comm d
      %%   LEFT JOIN exp2 ON exp2.algo = d.spec
      %% WHERE exp2.order != 0 AND nprocs < 2000
      %% GROUP BY exp2.order, exp2.name, nprocs, dton ORDER BY exp2.order, exp2.name, x, y
      \addplot coordinates { (320,97.6762) (480,97.699) (640,97.7308) (960,97.8218) (1280,97.9492) };
      \addlegendentry{name=\\kurpicz};
      \addplot coordinates { (320,410.658) (480,436.639) (640,459.493) (960,485.474) (1280,508.31) };
      \addlegendentry{name=\\hQuick};
      \addplot coordinates { (160,97.6746) (320,97.7326) (480,97.8345) (640,97.9901) (960,98.4412) (1280,99.123) };
      \addlegendentry{name=\\MSSimpleS};
      \addplot coordinates { (160,83.1361) (320,84.3505) (480,85.1018) (640,85.7024) (960,86.7635) (1280,87.8664) };
      \addlegendentry{name=\\MSLCPS};
      \addplot coordinates { (160,46.1562) (320,47.7708) (480,48.7615) (640,49.5343) (960,50.818) (1280,52.0315) };
      \addlegendentry{name=\\PDGolombS};
      \addplot coordinates { (160,55.1824) (320,56.9701) (480,58.0637) (640,58.8983) (960,60.2608) (1280,61.4904) };
      \addlegendentry{name=\\PDNoGolombS};

      \legend{}
    \end{semilogxaxis}
  \end{tikzpicture}

  \end{minipage}

  \medskip
  \ref{leg:dna_strong}
  \caption{\label{fig:crawlStrong}Running times and number of bytes sent per string in strong-scaling for \CommonCrawl (82\,GB) and \DnaReads (125\,GB) inputs.}

\end{figure*}
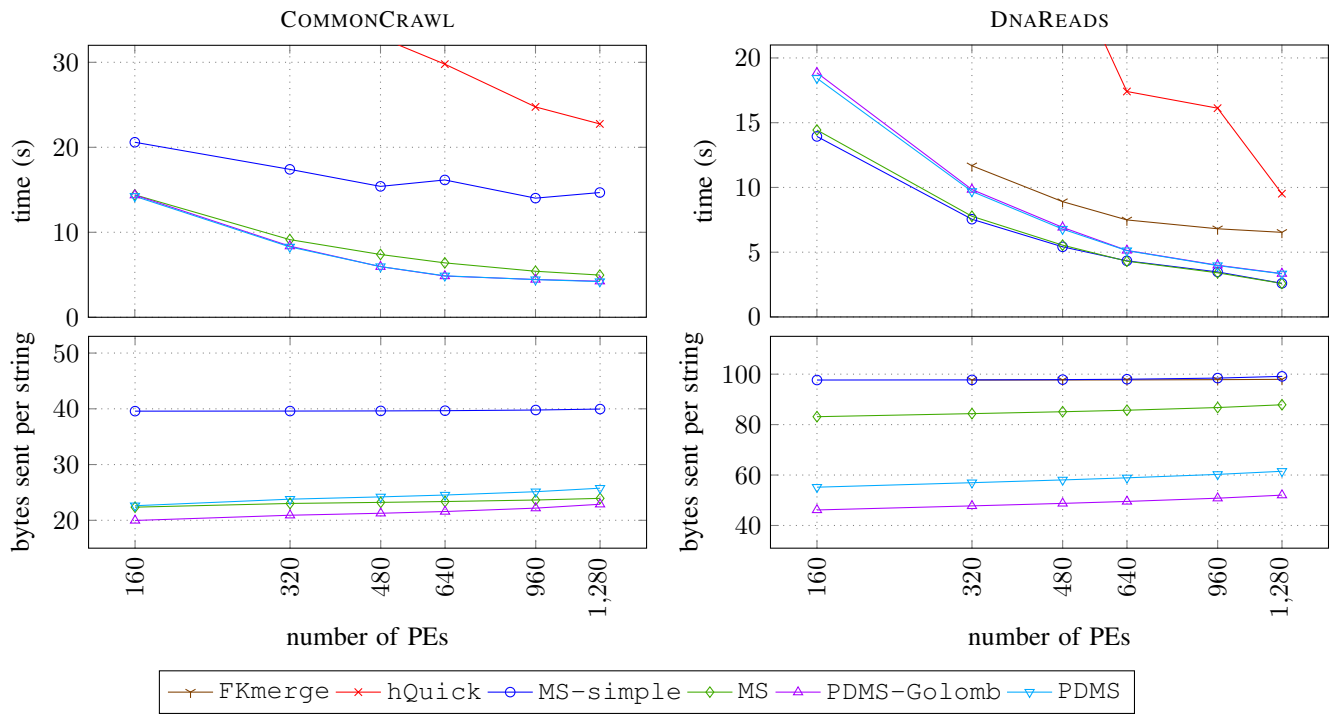

\subsection{Summary of Further Experiments}\label{ss:moreInputs}

We now outline the results of further experiments in \cite{Schimek19}.
Character-based sampling
is inconsequential for the \dton instances since their uniform length
and random distribution makes load balancing easy. For \CommonCrawl, our
initial implementation is detrimental indicating further improvement
potential. In particular, we have to carefully handle repeated short
substrings. A tendency of character-based sampling to select long
splitter keys indicates that one should perhaps consider using prefixes of
samples as splitters.

Besides the \CommonCrawl and \DnaReads instances, we also tried an instance consisting of
71\,GB of Wikipedia pages.  The results are similar to the
\CommonCrawl instance so that we do not show them here.

As a first attempt in the direction of suffix sorting, we considered
the first 3000 lines of the above Wikipedia instance as a single
string and used all their suffixes as input.  This instance has
$N\approx 104\cdot 10^9$ and $D\approx 10.4\cdot 10^6$, i.e.,
$D/N\approx 0.0001$. This is a very easy instance (execution time
about 0.2 seconds on 160 PEs) for algorithm \dPDSS and a fairly
difficult instance for all the other algorithms.  Algorithm \dPDSS is
about 30 times faster than the other algorithms for $p=160$. For
larger $p$ this advantage shrinks because larger inputs would be
needed to achieve scalability. However, these inputs would be very
expensive for the other algorithms so that we did not pursue suffix
instances given our limited compute resources.

We also generated \emph{skewed} variants of our \dton-instances as
follows: The 20\,\% smallest of these strings are padded with additional
characters that make them 4 times longer (now 2000 characters) but
without contributing to the distinguishing prefixes. The relative ranking of the algorithms
\hQuick, \kurpicz, \dMSS, and \dPDSS remains the same.
Among the variants of \dMSS, those with character-based sampling now profit
because they avoid deteriorating load balance due to the skewed distribution of output string lengths.

%%%%%%%%%%%%%%%%%%%%%%%%%%%%%%%%%%%%%%%%%%%%%%%%%%%%%%%%%%%%%%%%%%%%%%
\section{Conclusions and Future Work}\label{s:conclusion}

With Algorithm~\dPDSS, we have developed a distributed-memory string
sorting algorithm that efficiently sorts large data sets both in
theory and practice.  The algorithm is several times faster than the
best previous algorithms and scales well for sufficiently large
inputs.

One approach to further optimize the algorithm is
to improve splitter selection.  Analogous to the work done for atomic
sorting \cite{axtmann2017robust} one could remove load balancing
problems due to duplicate strings by tie breaking techniques. One
could also consider whether it helps to look for short splitter
strings. The bounds could also be improved by going from deterministic
sampling to random sampling. This requires less samples and, in
expectation, the sample strings have average length rather than
length $\lmax$. Adapting the techniques from
\cite{VarEtAl91,axtmann2015practical} for \emph{perfect} splitting in
atomic sorting to string sorting could also be interesting. Probably
this only makes sense if we also use a refined cost model that takes
both the number of strings and their distinguishing prefix lengths
into account.

To speed up sorting of the sample but also for other applications, it
is interesting to look for parallel string sorting algorithms for
small inputs that are faster than \hQuick. One approach would be to
adapt the key idea of Multikey Quicksort to look not at entire strings
as pivots \cite{bentley1997fast}. In \cite{bingmann2018scalable} this
is refined by looking at several characters at once. Probably for a
distributed algorithm one should look at up to $\Oh{\alpha\log p/(\beta\log\sigma)}$
characters at a time to find the right balance between latency and
communication volume.

An interesting observation is that algorithms based on data
partitioning rather than merging are successful both for atomic
sorting \cite{axtmann2017robust} and shared-memory string sorting
\cite{bingmann2018scalable}. We have chosen a merging based algorithm
here since it is not clear how to do LCP compression without locally
sorting first.

The large $D/N$ values in our practical inputs are in part due to many
repeated keys.  Perhaps one could design string sorting algorithms
that do not communicate duplicate keys. One could modify
Algorithm~\dPDSS so that it detects likely duplicates and decides to
communicate only one copy of each duplicate to their common
destination PE. A problem with this approach is that we cannot
guarantee that these strings are duplicates.  We could enforce a very
small false positive rate but we would end up with a result that is
only approximately sorted.

The average-case upper bound of $\Oh{\log p}$ bits per string from
Section~\ref{ss:average} is intriguing from a theoretical point of
view. For atomic sorting, the average case and the worst case running
time share the same upper and lower bounds.  Does this extend to the
communication complexity of string sorting? Are there algorithms that
need $o(D/n)$ communication volume in the worst case?  What are the
lower bounds? $D/n$? $\log p$? Something in between?  The answer will
likely depend on the small print in how we define our sorting problem.
It should also be noted here that the lower bound for the easier
problem of duplicate detection is conjectured to be $\log p$
\cite{sanders2013communication} but that this is also still an open
problem for distributed communication complexity with point-to-point
communication.

The algorithm for approximating distinguishing prefixes from
Section~\ref{ss:distinguishing} is an overkill if we only need
information on global values like $D/n$ or its variance.  These values
can be approximated more efficiently by sampling. A simple approach is
to gossip a small sample of the input strings. Then, without further
communication, their distinguishing prefix sizes can be computed
locally.  However, this way we can only process small samples which
might be insufficient when $\dmax\gg D/n$ -- a small sample
is insufficient if $D/n$ is dominated by a small number of strings
with very large $\dpre{s}$.%
\footnote{Initial calculations indicate that a sample of size
  $\Th{\varepsilon^{-2}n\dmax/D}$ is needed to approximate $D$ with
  standard deviation $\varepsilon D$.}  More efficiently, we can take a
Bernoulli sample of prefixes of \emph{keys} rather than input
strings. This allows us to still use distributed hashing and thus
makes the algorithm more scalable. Also this might reduce the amount
of local work.

\section*{Acknowledgment}

We used the ForHLR~I cluster funded by the Ministry of Science, Research and the Arts Baden-W\"urttemberg and by the Federal Ministry of Education and Research.



%\bibliographystyle{IEEEtran}

\bibliography{references}

\end{document}

%% file: makros.tex
\usepackage{amsfonts}

\newcommand{\ceil}[1]{\left\lceil #1\right\rceil}
\newcommand{\floor}[1]{\left\lfloor #1\right\rfloor}
\newcommand{\abs}[1]{\left| #1\right|}

\newcommand{\realrange}[2]{\left[#1, #2\right]}

\newcommand{\unitrange}[2]{\realrange{0}{1}}

\newcommand{\Oh}[1]{\mathcal{O}({#1})}

\newcommand{\Th}[1]{\Theta({#1})}

\newcommand{\Om}[1]{\Omega({#1})}

\newcommand{\llabel}[1]{\label{\labelprefix:#1}}
\newcommand{\labelprefix}{} % later redefined using renewcommand

\newcommand{\discussionsize}{\small}

\newenvironment{code}{\noindent%\sf%
\begin{tabbing}%
\hspace{2em}\=\hspace{2em}\=\hspace{2em}\=\hspace{2em}\=\hspace{2em}\=%
\hspace{2em}\=\hspace{2em}\=\hspace{2em}\=\hspace{2em}\=\hspace{2em}\=%
\kill}{\end{tabbing}}

\newcommand{\labelcommand}{}
\newcommand{\captiontext}{}
\newsavebox{\codeparam}
\newcounter{lineNumber}
\newenvironment{disscodepos}[3]{%
\renewcommand{\labelcommand}{#2}%
\renewcommand{\captiontext}{#3}%
\sbox{\codeparam}{\parbox{\textwidth}{#3}}%
\begin{figure}[#1]\begin{center}\begin{code}\setcounter{lineNumber}{1}}{%
\end{code}\end{center}\caption{\llabel{\labelcommand}\captiontext}\end{figure}}

{\end{disscodepos}}

\newdimen\endofsize\endofsize=0.5em
\def\endofbeweis{~\quad\hglue\hsize minus\hsize
                 \hbox{\vrule height \endofsize width
\endofsize}\par}